\RequirePackage{fix-cm}
\documentclass[12pt]{amsart}
\usepackage{amsmath,amssymb}
\usepackage{graphicx}
\usepackage[all,cmtip]{xy}
\usepackage{mathptmx}      
\usepackage{latexsym}

\newtheorem{proposition}{Proposition}
\newtheorem{theorem}{Theorem}
\newtheorem{definition}{Definition}
\newtheorem{corollary}{Corollary}
\newtheorem{lemma}{Lemma}
\newtheorem{remark}{Remark}
\newtheorem{ej}{Example}
\newcommand{\Me}{\mathcal{M}}
\newcommand{\Se}{\mathcal{S}}
\newcommand{\He}{\mathcal{H}}
\newcommand{\se}{{\bf S}}

\begin{document}

\title{No-Arbitrage Symmetries}

\author{I. L. Degano}
\address[I.L. Degano]{Department of Mathematics\\
Mar del Plata National University\\
Mar del Plata, Funes 3350, 7600 Argentina}
\email{ivandegano@mdp.edu.ar}
\author{S.E.Ferrando}
\address[S.E.Ferrando]{Department of Mathematics\\
Ryerson University\\
350 Victoria Street}
\email{ferrando@ryerson.ca}
\author{A.L. Gonzalez}
\address[A.L. Gonzalez]{Department of Mathematics\\
Mar del Plata National University\\
Mar del Plata, Funes 3350, 7600 Argentina}
\email{algonzal@mdp.edu.ar}


\date{August 12, 2020}


\maketitle

\begin{abstract}
{The no-arbitrage property is widely accepted to be a centerpiece of modern financial mathematics and could be considered to be a financial law applicable to a large class of (idealized) markets. The paper addresses the following basic question: can one characterize the class of transformations that leave the law of no-arbitrage invariant? We provide a geometric formalization of this question  in a non probabilistic setting of discrete time, the so-called trajectorial models. The paper then characterizes, in a local sense, the no-arbitrage symmetries and illustrates their meaning in a detailed example. Our context makes the result available to the stochastic setting as a special case.}
\end{abstract}

\vspace{.1in}
\noindent

\vspace{.1in}
\noindent
2000 Math Subject Classification: 91B24; 91G10; 26B25; 49J35.

\section{Introduction}
The principle of no-arbitrage plays a fundamental role in modern financial mathematics, see \cite{follmer} and references therein (we mostly restrict our comments and developments to a discrete time setting). In plain language, the assumption of no-arbitrage  means that risky asset models should rule out apriori the possibility of making a profit without taking in any risk. This hypothesis implies a pricing methodology based on martingale stochastic processes, this is the risk-neutral valuation (\cite{bingham}), and as such plays the role of a financial law. Its empirical validity has been studied (\cite{kamara}) and even if arbitrage opportunities may be available they are believed to be rare, short lived and hard to profit from. One could compare the notion of no-arbitrage to a physical law that applies in idealized conditions such as the principle of inertia and think of it as a fundamental financial law applicable to a large class of (idealized) markets. With this point of view in mind, and as a preliminary step, we pose the question: {\it can we characterize the class of transformations that leave the law of no-arbitrage invariant}? This is  much akin to Galilean/Lorentz transformations leaving the class of inertial frames invariant. Therefore, we look for no-arbitrage preserving transformations (referred also as no-arbitrage symmetries) mapping a set of  financial events
into another set  of financial events with the property that the no-arbitrage
property holds for both classes of events. One is also interested in providing a financial interpretation to such set of symmetries (much like Galilean/Lorentz transformations having a physical interpretation) and exploring some financial implications.

The financial events mentioned in the previous paragraph have to be linked to financial
transactions as it is the latter that fall under the scope of the no-arbitrage principle. In most situations, each such transaction involves two goods $X$ and $Y$ and  a {\it price} $X_Y(t)$ so that $X = X_Y(t)~ Y$. $X_Y(t)$ is the integer
number of units of asset $Y$ required to purchase one unit of asset $X$.
In terms of dimensional units  $[X_Y(t)]= [X]/[Y]$ and the reference asset $Y$ is called the (chosen) numeraire (\cite{vecer}). This discussion suggests that an analysis of the no-arbitrage principle could be done in terms of prices and numeraires (numeraire free approaches are also possible) and that is the way we proceed in the paper.

The setting where we precisely pose and answer the above raised question is a set  of sequences of  multidimensional  prices that evolve in discrete time.
That type of set is called a trajectorial model (for a set of risky assets); in our investigation there is no need to assume any probability structure on such a set. In this way we can work on a more general setting than the (discrete time) stochastic framework and the question we study becomes a natural one unencumbered by unnecessary additional structure. We also briefly indicate how our main results, characterizing no-arbitrage transformations, apply to the stochastic setting as a special case.

The mentioned trajectorial framework has been developed in \cite{ferrando} and \cite{degano} (see also \cite{ferrando2}) for the $1$-dimensional case  with a $0$ interest rate bank account as numeraire and is here extended to the $d$-dimensional case and a general numeraire (but we restrict ourselves to finite time, as opposed to unbounded or infinite discrete time as in \cite{ferrando} and \cite{ferrando2}, respectively). It then follows that the global notion of no-arbitrage, i.e. involving several time steps, can be reduced to the one step notion of no-arbitrage. This is a classical reduction in discrete and finite time and allows to concentrate our efforts in the local notion (i.e. involving one step into the future) of no-arbitrage.

We can now be more precise about our search for no-arbitrage preserving transformations.

\vspace{.1in}
\noindent
{\bf Definition} [Disperse sets]  
Consider a set $E \subseteq \mathbb{R}^d$; $E$ is called {\it disperse} if for each  $h \in \mathbb{R}^d$:
$    [h \cdot Y= 0~\forall ~Y \in E]~\mbox{or}~~[\inf_{Y \in E}  h \cdot Y < 0  ~~\&~~~\sup_{Y \in E}  h \cdot Y > 0]$,
where $h \cdot Y$ represents the euclidean inner product.

\vspace{.1in}
We prove in the paper (by means of  Proposition \ref{localDefinitions}, Proposition \ref{suffCondForArb-freeAnd0Neutral} and Definition \ref{trSpaceLocArbFree})
that the notion of a set being disperse is equivalent to the risky assets (one step time evolution) obeying  the no-arbitrage principle. Therefore, our original question on no-arbitrage preserving transformations becomes: characterize the set of transformations that leave the disperse property invariant. We are then in the context of the modern view of a geometry where we study the set of transformations that leave certain properties of a space of points invariant.

We work on a self contained framework for financial markets centered on a set of (mutidimensional) trajectories modeling a collection of assets with the possibility that any of them could play the role of a numeraire asset. No probability measures, filtrations, cardinality or topological assumptions are required of the trajectory set. The approach singles out local trajectory properties that can be used to consistently build an associated option price theory. The paper does not pursue this latter possibility as we focus on symmetry transformations (option pricing developments are in \cite{ferrando} for the $1$-dimensional case). Such trajectory properties have already made their appearance in the stochastic literature (\cite{dalang}, \cite{bender2}, \cite{jarrow}). To relate to the well established stochastic approach (\cite{follmer}) the reader could think that our paper concentrates on financial developments that only depend on the support of a given stochastic process independently of any possible probability distribution. Some connections with the stochastic setting are developed in \cite{ferrando},
the reference \cite{ferrando2} provides a first mathematical step
to extend some martingale notions from the standard setting to a trajectorial setting.

The subject of the paper is the study of fundamental symmetry transformations associated to the no-arbitrage principle in a non-probabilistic setting. Results are obtained with minimal assumptions and, in this way, providing a wider financial context for their availability.

We describe the contents of the paper.
Section \ref{sec:setting} introduces the setting which is centered on a trajectory space. Section \ref{sec:arby0neutral} studies the notions of no-arbitrage and $0$-neutrality (a weakening of no-arbitrage) in trajectory based markets. Section \ref{global&local} introduces local conditions (i.e. properties that are conditioned on a given state of affairs and involve one step into the future)  which are necessary  and sufficient to establish no-arbitrage and $0$-neutrality. These local conditions play the analogue role of the martingale condition in stochastic markets. Section \ref{geometricCharacterizations} develops purely geometric results in $\mathbb{R}^d$, independent of any financial setting, that form the backbone to derive the set of  no-arbitrage and $0$-neutral symmetries. 
Section \ref{sec:invarianceOfLocalProperties} characterizes two classes of transformations, one preserving  the local no-arbitrage property and the other class preserving the $0$-neutral property.  In particular, we prove that a change of numeraire belongs to both such classes of transformations. The uncovered transformations
can then be considered to be symmetries satisfied by price relationships and we provide an illustrative example.
Appendix \ref{proofsAndResults} provides proofs for results in the main body of the paper. Appendix \ref{app:conv} develops some results on convex analysis that we rely upon. Finally, we use the words arbitrage-free and no-arbitrage interchangeably.

\section{General Trajectorial Setting} \label{sec:setting}

We introduce the mathematical setting of a dynamic financial market with a finite number of assets whose initial prices are known. Uncertainty of future prices is given by a set of multidimensional sequences that we call trajectories. The trading strategies are given by portfolios that will be successively re-adjusted, taking into account the information available at each stage. The present paper is essentially self-contained We extend work presented in  \cite{degano} and \cite{ferrando}. The latter reference presents a non-probabilistic one-dimensional, discrete time, setting to price European options. The reference \cite{degano} provides examples and a computational algorithm to evaluate price bounds for European options. A detailed discussion and justification of why a trajectorial modeling approach is
worth studying is presented in \cite{ferrando} as well as in \cite{degano}.   The present paper extends the setting from those two papers to the multidimensional case.

There is empirical evidence suggesting that liquid markets do not allow for arbitrage opportunities. Therefore, and from a modeling point of view,
the {\it no-arbitrage principle} assumes that market models should not contain arbitrage strategies. The no-arbitrage assumption allows to develop a theory constraining relative prices. We remark in passing that under the weaker condition of $0$-neutrality (see Definition \ref{0Neutral} as well as \cite{ferrando}), it is possible to obtain well defined price bounds for European options.

More precisely, we consider a market with $d+1$ assets that evolve in a fixed time interval $[0,T]$. The model will be discrete in the sense that the trading instances are indexed by integer numbers. Given $s_0=(s_0^0,s_0^1,s_0^2,\dots,s_0^d) \in \mathbb{R}^{d+1}$, as initial prices of assets $S^k$, we will denote by a trajectory $S$, a sequence taking values in $\mathbb{R}^{d+1}$ such that
\[ S_i=(S^0_i,S^1_i,S^2_i,\dots,S^d_i)~~ \mbox{with}~~ S_0=s_0. \]

A portfolio will be a sequence of functions defined on the trajectory sets which we will denote by
\[ \Phi=(H^0,H)=\{(H^0_i,H^1_i, \dots, H^d_i)\}_{i \ge 0} .\]
Each coordinate $H^j_i$, $0 \le j \le d$ represents the portfolio holdings at stage $i$, for the $j$-th asset with $[H^j_i]= {\bf 1}_{S^j}$ (a unit of asset $S^j$). The asset values and the invested amounts can take values in general subsets of the real numbers.

The portfolio re-balancing stages may be triggered by arbitrary events of the market without the need to be directly associated with time. To incorporate this greater degree of generality we will add a new source of uncertainty to the trajectories' coordinates (these additional coordinates are relevant when constructing specific models). We will denote them by $W = \{W_i\}_{i \ge 0}$, the $W_i$ can be vector valued and take values in arbitrary sets. In financial terms, this new variable can represent any observable value of interest, such as volume of transactions, time, quadratic variation of trajectories, etc., as in  \cite{ferrando19}.

In case one intends to price financial derivatives in the proposed setting, we add a finite time horizon $T$.  We will  use a positive integer $m$, to indicate the stage at which the trajectory reaches the time $T$. This new variable plays a key role in calculating the fair price interval for options, although it does not intervene in the market properties.

\begin{definition}[Trajectory set]  \label{trajectories}
Consider  $\Sigma =\{\Sigma_i\}$  a given family of subsets of $\mathbb{R}^{d+1}$, $\Omega=\{ \Omega_i \}$ is a family of sets and $\Theta \subseteq \mathbb{N}$.
For given $s_0 \in \mathbb{R}^{d+1}$ and $w_0 \in \Omega_0$, a \emph{trajectory based set} $\mathcal{S}$ is a subset of
\begin{equation}  \nonumber
\mathcal{S}_{\infty}(s_0,w_0) \equiv \left\lbrace \se = \{{\bf S}_i \equiv (S_i,W_i,m)\}_{i \geq 0}: ~S_i \in \Sigma_i, ~ W_i \in \Omega_i,~m \in \Theta \right\rbrace,
\end{equation}
such that $(S_0,W_0)= (s_0,w_0)$. The elements of $\mathcal{S}$  will be called \emph{trajectories}.
\end{definition}

It is important to note that if $\tilde{{\bf S}} = \{(\tilde{S}_i, \tilde{W}_i, \tilde{m}) \} $ and $ \hat{{\bf S}} = \{(\hat{S}_i, \hat{W}_i, \hat{m}) \} $ are two trajectories, $ \tilde{\se}_i $ could unfold at a different time than $ \hat{\se}_i $. That is, the index $i$ will be associated with portfolio re-balances stages but they will not be necessarily associated to (uniform) time. It is only assumed that the stage $i + 1$ occurs temporarily after the  stage $ i $.

We define $M: \Se \rightarrow \mathbb{N} $ as the projection on the third coordinate of $ \se $, that is: $ M(\se) = m $. The results and properties that appear in this section only involve the first coordinate $ S_i $ nonetheless, we will continue using the notation that includes the coordinates $W_i$ for consistency.

To build an adequate market model, we are going to require that any portfolio be non-anticipative. The non-anticipativity of the portfolios expresses the fact that investments must be made at the beginning of each period, so that they can not anticipate specific future price changes.

\begin{definition}[Portfolio]  \label{locallyDefinedPortfolios}
Let $\Se$ be a trajectory  set, a \emph{portfolio} $\Phi$ is a sequence of (pairs of) functions $\Phi \equiv \{(H^0_i, H_i)\}_{i\geq 0}$ with
$H^0_i:\mathcal{S} \rightarrow \mathbb{R}$ and $H_i: \mathcal{S} \rightarrow \mathbb{R}^d$ such that for all $\se, \se' \in \mathcal{S}$, with $\se'_i = \se_i$ for all $0 \leq i \leq k$, where $k < \min\{M(\se),M(\se')\}$, then $\Phi_k(\se) = \Phi_k(\se')$.
\end{definition}

For a portfolio $ \Phi $, $ H^j_i (\se) $ represents the number of units held for the $j$-th asset during the period between $ i $ and $ i + 1 $. Therefore, $ H^j_i(\se)~ S ^ j_i $ is the {\it value}  invested in the $ j $-th asset at stage $ i $, while $ H^j_i(\se)~ S^ j_{i + 1} $ is the value  just before rebalancing at the end of the period. So, the total value of the portfolio $ \Phi $ at the beginning of the period $ i  $ is
\begin{equation} \nonumber
H^0_i (\se) S^ 0_i + H_i (\se) \cdot S_i 
\equiv H^0_i (\se) S^0_i + \sum_{j = 1}^d H^j_i(\se)~ S^ j_i,
\end{equation}
and at the end of the period, the value of $ \Phi $ will change to
\begin{equation} \nonumber
H^0_i (\se) S^0_{i + 1} + H_i(\se) \cdot S_{i + 1} = H^0_i(\se)~ S^0_{i + 1} + \sum_{j = 1}^d H^j_i(\se)~ S^j_{i + 1}.
\end{equation}

In the next re-balancing, the investor will invest $\Phi_{i + 1}$; in general, $ H^0_{i + 1}(\se)~ S^0_{i + 1} + H_{i + 1}(\se) \cdot S_{i + 1} $ may be different from $ H^0_{i}(\se) ~S^0_{i + 1} + H_{i}(\se) \cdot S_{i + 1} $. In this latter case, it follows that some units of the assets were added or removed, without replacement, from the portfolio. However, this situation is precluded for many applications. For example, if the goal is to look for a ``fair'' price for a certain financial contract, this value should be the minimum necessary to cover the obligations generated by the contract, that is, any injection or withdrawal of money will affect this property. This reasoning justifies the use of the following concept.

\begin{definition}[Self-financing portfolio]
A portfolio $\Phi$ is called \emph{self-financing} if for all $\se \in \mathcal{S}$ and $i \geq 0$,
\begin{equation} \label{selfFinancing}
H^0_{i}(\se)~S^0_{i+1}+ H_{i}(\se) \cdot S_{i+1} = H^0_{i+1}(\se)~S^0_{i+1}+ H_{i+1}(\se) \cdot S_{i+1}.
\end{equation}
\end{definition}

The self-financing property means that the portfolio is re-balanced in such a way that its value is preserved. From this property it is clear that the accumulated gains and losses resulting from price fluctuations are the only sources of variation of the portfolio:
\begin{equation}  \label{portfolioValue}
H^0_{k}(\se)~S^0_{k}+ H_{k}(\se) \cdot S_{k} = H^0_{0}~S^0_{0}+ H_{0} \cdot S_{0} + \sum_{i=0}^{k-1} \left( H^0_{i}(\se)~\Delta_i S^0 + H_{i}(\se) \cdot \Delta_i S \right),
\end{equation}
for $ k \ge 0 $, where $ \Delta_i S ^ 0 = S ^ 0_ {i + 1} -S ^ 0_i $ and $ \Delta_i S = S_ {i + 1} -S_i $. The value $ H^0_{0}~ S ^ 0_{0} + H_{0} \cdot S_{0} $ represents the initial investment corresponding to the portfolio coordinate $\Phi_0 $.

We will mention below some examples of strategies that will be used later.

\begin{ej} \label{ej:estrategias}
\begin{enumerate}
\item The null portfolio $\Phi=0$,
\[ 0(\se)=\{(0,{\bf 0})\}_{i \ge 0} \mbox{ for all } \se \in \Se \]
where ${\bf 0}$ is the null vector of $\mathbb{R}^d$.
\item Set $h \in \mathbb{R}$ and ${\bf h} \in \mathbb{R}^d$, we will define by constant portfolio $\Phi={\bf h}$ by
\[ {\bf h}(\se)=\{(h,{\bf h})\}_{i \ge 0} \mbox{ for all } \se \in \Se. \]
\item Set $\Phi = \{(H^0_i, H_i)\}_{i\geq 0}$ a self-financing portfolio. We will denote by $-\Phi$ to the sequence of functions $\{(-H^0_i, -H_i)\}_{i\geq 0}$. It is easy to see that $-\Phi$ is a  self-financing portfolio.

\item Set $\hat{\Phi}=\{(\hat{H}^0_i, \hat{H}_i)\}_{i\geq 0}$ and $\tilde{\Phi}=\{(\tilde{H}^0_i, \tilde{H}_i)\}_{i\geq 0}$ two portfolios. We define $\Phi \equiv \hat{\Phi} + \tilde{\Phi}$ to be the sequence
\[ \Phi= \{ \hat{H}^0_i+ \tilde{H}^0_i,\hat{H}_i+\tilde{H}_i \}_{i \ge 0}.\]
\end{enumerate}
\end{ej}

\subsection{Numeraire}
\label{currencyAndNumeraire}

\vspace{.1in}
To be definite, we will consider that the real numbers $S^k_i$ express the price of asset $S^k$ in a common currency, a unit of which we denote generically by $\$$. That is, in terms of dimensions $[S^k_i]= \$/{\bf 1}_{S^k}$ where ${\bf 1}_{S^k}$ is one unit of asset $S^k$ (it is well known that an algebra of dimensions is available through dimensional analysis as in \cite{whitney}).
Notice that $[H^k_i]={\bf 1}_{S^k}$.
 On the other hand, for financial reasons, it is important to work with an arbitrary reference asset; this is achieved by taking a reference asset as \emph{num\'eraire}. For example, in some cases it is useful to select the value of a bank account as num\'eraire.

Toward this end, we will assume from here onward that $ S ^ 0_i> 0 $ for all $ i \ge 0 $. This hypothesis will allow us to use $ S^0 $ as num\'eraire. For each $ \se \in \Se $, we will build a sequence of {\it relative prices} ${\bf X}(\se) = \{(X(S_i), W_i, m) \} _ {i \ge 0} $ where $X: D \subseteq \mathbb{R}^{d + 1} \rightarrow \mathbb{R}^ d $ is a {\it perspective} function defined by
\begin{equation} \label{eqn:X}
X(s) \equiv \left(\frac{s^{1}}{s^0},\dots,\frac{s^{d}}{s^0}\right), \quad D\equiv\{ s=(s^0,\dots,s^d) \in \mathbb{R}^{d+1}:s^0>0\}.
\end{equation}
The numerical value of $ X^j(S_i)$ (i.e. stripped from its units), is the number of units of the asset $S^0$, now the num\'eraire, which are required to acquire one unit of the $S^j$ asset.
\begin{remark}
Notice the above definition of $X$ singles out $s^0$ but of course any other coordinate could be used (relying on the $0$-component simplifies the notation). In fact, and for more generality, one could replace $s^0$ by a linear map $B(s)>0$ on $s^k >0$. We do not pursue here this possibility but our results will apply to such numeraire by just moving to a new trajectory market with $s^0= B(s)$.
\end{remark}

Given $ \se \in \Se $ and $ k \ge 0 $, we will denote by $ V^{\Phi}_k (\se) $ the relative value of the portfolio $ \Phi \in \He $ given by
\begin{equation} \nonumber
V^{\Phi}_k(\se) \equiv H^0_k(\se)+H_k(\se) \cdot X(S_k).
\end{equation}
Clearly $ V^{\Phi}_k (\se) = \frac{\Phi_k(\se) \cdot S_k} {S^0_k} $,
then $ V^{\Phi}_k (\se) $ can be interpreted as the value of the portfolio at the beginning of the stage $ k $ expressed in units of the num\'eraire. In addition,  $ G^{\Phi}_k (\se)$ will denote the profits generated up to the stage $ k $ associated with $ \Phi \in \He $ for a trajectory $ \se \in \Se $, that is
\begin{equation} \label{eqn:G}
G^{\Phi}_k(\se) \equiv  \sum_{i=0}^{k-1} H_i(\se) \cdot \Delta_i X(S) ~ \mbox{for}~~ k \ge 0 ~~\mbox{where}~~ \Delta_i X(S) = X(S_{i + 1}) - X(S_i).
\end{equation}
 $ G^{\Phi}_k (\se) $ reflects, in terms of the num\'eraire, the net gains accumulated by the portfolio $\Phi$ at the beginning of the $ k $-th stage.

A self-financing portfolio for a path $\se \in \Se $ will also be self-financing for the $ X(\se) $ sequence.

\begin{proposition} \label{prop:auto}
  Let $ \Se $ be a space of trajectories, and let $ \Phi $ be a portfolio on $\Se$. Then the following statements are equivalent:
  \begin{enumerate}
    \item $\Phi$ is self-financing.
    \item $H^0_{i-1}(\se) + H_{i-1}(\se) \cdot X(S_i)=H^0_i(\se) + H_{i}(\se) \cdot X(S_i)$ for all $\se \in \Se$ and $i \ge 0$.
    \item $\displaystyle V^{\Phi}_k(\se)=V^{\Phi}_0+G^{\Phi}_k(\se)=H^0_0 + H_0 \cdot X(S_0) + \sum_{i=0}^{k-1} H_i(\se) \cdot \Delta_i X(S)$ for all $k \ge 0$.
  \end{enumerate}
\end{proposition}
\begin{proof}
Note that Proposition 5.7 of \cite{follmer} is valid even in cases where we do not have a market indexed by pre-set time stages. Therefore the same idea used in that result applies to our setting.
\end{proof}

\begin{remark}\label{H^0}
From the previous Proposition, we know that the $H^0$ component of a self-financed portfolio $ \Phi $ satisfies
  \begin{equation} \label{eqn:defH1}
  H^0_k(\se)-H^0_{k-1}(\se)=-(H_{k}(\se)-H_{k-1}(\se)) \cdot X(S_k).
  \end{equation}
  Given that
  \begin{equation} \label{eqn:defH2}
  H^0_0=V^{\Phi}_0-H_0\cdot X(S_0),
  \end{equation}
  the sequence $ H^0 $ is completely determined by the initial investment $ V^{\Phi}_0 $ and $ H $ by means of the previous equations.
\end{remark}
\begin{remark}\label{H simplificado} For a given set of portfolios
$\He$, in virtue of Remark \ref{H^0}, and display (\ref{eqn:G}) (which depends on $\Phi=(H^0,H)$, just through $H$) we will set the definition
\begin{equation}  \label{displaySoDefinitionIsNotLost}
\He_{\Se}\equiv\{H: (H^0,H)\in\He\}
\end{equation}
for later use.
\end{remark}

\begin{definition}[Trajectory  market]  \label{discreteMarkets}
Given $s_0 \in \mathbb{R}^{d+1}$, $w_0 \in \Omega_0$, a trajectory based set $\mathcal{S} \subseteq \mathcal{S}_{\infty}(s_0,w_0)$ and a portfolio set $\mathcal{H}$, we say that $\mathcal{M} = \mathcal{S} \times \mathcal{H}$ is a \emph{trajectory based market} if it satisfies the following properties:
\begin{enumerate}
\item For each $\se \in \Se$, the coordinate $S^0_i >0$ for all $i \ge 0$.
\item  All $\Phi \in \He$ are self-financing and $\Phi=0$ belongs to $\mathcal{H}$.
\item For each $(\se, \Phi) \in \mathcal{M}$ there exists $N_{\Phi}(\se) \in \mathbb{N}$ such that $\Phi_{k}(\se) = \Phi_{N_{\Phi}}(\se) =0$ for all $k \geq N_{\Phi}(\se)$.
\end{enumerate}
We will say that the market is \emph{semi-bounded} if for each $ \Phi \in \He $ there is $ n_ {\Phi} \in \mathbb{N} $ such that $ N_ {\Phi}(\se) \le n_{\Phi} $ for all $ \se \in \Se $ and it is \emph{$ n $-bounded}, for $ n \in \mathbb{N} $, if $ N_{\Phi} (\se) \le M (\se) \le n $ for each pair $ (\se, \Phi) \in \Me $. A portfolio set $\mathcal{H}$ obeying items $(2)$ and $(3)$ above will be called {\it admissible}.
\end{definition}

The third property of the previous definition states that the adjustments of the portfolio $\Phi$ for a trajectory $ \se $ will end at, or before, the stage $ N_{\Phi} (\se)-1 $, which means that the portfolio was liquidated on, or before, the period $ N_{\Phi} (\se) $. In this case, the corresponding portfolio will be called {\it liquidated}.

The above setting incorporates, as a special case, a discrete time stochastic model.  Given a process $Y= \{Y_i= (Y_i^0, \ldots, Y_i^d)\}_{i \geq 0}$ on a probability space $(\Omega, P)$
with filtration $\mathcal{F}= \{\mathcal{F}_i\}_{i \geq 0}$ and $\mathcal{F}_0$ trivial,
$Y_i^k: \Omega \rightarrow \mathbb{R}$, $Y_i^k \in \mathcal{F}_i$.
We can then define $\mathcal{S}= \{{\bf S} = \{(S_i \equiv Y_i(\omega))\}_{i \geq 0}: \mbox{for some }~~\omega  \in \Omega\}$. One can also define
a set of trajectories $\mathcal{S}$ by means of
a sequence of admissible stopping times $\tau = \{\tau_i\}_{i \geq 0}$: ${\bf S} \in \mathcal{S}$ then ${\bf S} = \{(S_i \equiv Y_{\tau_i(\omega)}(\omega))\}_{i \geq 0}$ for some $\omega \in \Omega$. Another way to proceed is to use a given collection of such sequences of stopping times, for the details we refer to Section 6 in \cite{ferrando}.

\section{Arbitrage and 0-Neutrality} \label{sec:arby0neutral}

A model for common market situations should not allow for investors that are able to generate a profit in a transaction without any risk/possibility of losing money.  Such an investment opportunity is called an arbitrage opportunity.

\begin{definition}[Arbitrage opportunity] \label{ArbitrageDefinition}
Given a trajectory based market $\mathcal{M}=\Se \times \He$, $\Phi \in \mathcal{H}$ is an \emph{arbitrage opportunity} if:
\begin{itemize}
\item $\forall \se \in \mathcal{S}$,  $V^{\Phi}_{N_{\Phi}}(\se) \geq V^{\Phi}_0$.
\item $\exists \se^{\ast} \in \Se$ such that $V^{\Phi}_{N_{\Phi}}(\se^{\ast}) > V^{\Phi}_0$.
\end{itemize}
We say that $\mathcal{M}$ is \emph{arbitrage-free} if $\mathcal{H}$ does not contain arbitrage opportunities.
\end{definition}

The particular case $S_i^0=1$, for all $i$, gives $X(S_i)=(S^1_i,\ldots,S^d_i)$; i.e. the original currency is the asset $S^0$ and is being used as  numeraire and so $[S_i^0] = \$/\$$. Currency, if included as a traded asset and in the presence of a riskless bank account with non-zero interest rates, will lead to an arbitrage as per Definition \ref{ArbitrageDefinition}. That is, currency, under the mentioned conditions will be banned as a traded asset whenever we assume a no arbitrage market (as well as a $0$-neutral market). Notice the relevant discussion in \cite{vecer} about arbitrage and non-arbitrage assets, currency being an arbitrage asset (as contrasted to an interest bearing money market account).

Our use of an arbitrary value for $V^{\Phi}_0$ in the definition of an arbitrage opportunity is nonstandard, textbook definitions require $V^{\Phi}_0 \leq 0$ (see \cite{follmer}).
One can see that the existence of an arbitrage as per Definition \ref{ArbitrageDefinition} is equivalent to the existence of an
arbitrage portfolio $\tilde{\Phi}$ with $V^{\tilde{\Phi}}_0= 0$ and so proving the equivalence of our definition and the standard definition. This equivalence allows also to show that Definition \ref{ArbitrageDefinition} is  invariant under a change of numeraire and so the latter transformation will be a no arbitrage symmetry according to our definitions, this we show explicitly in Corollary \ref{cor:numeraire}.

The arbitrage-free condition is sufficient for the model to provide fair option prices (a well known result in the classical financial literature.)
One can relax the arbitrage free criteria to the requirement that the largest of the minimum possible gains that can be obtained by means of the strategies available in the market is $0$. This notion was originally presented in \cite{BJN} (as equivalent with arbitrage-free) and then formally defined and clarified for the case of a single risky asset in \cite{ferrando} and \cite{degano}.

\begin{definition}[$0$-neutral market]  \label{0Neutral}
Let $\Me=\Se \times \He$ be a trajectory based market. We say that $\mathcal{M}$ is \emph{0-neutral} if
\begin{equation} \nonumber
  \sup_{\Phi \in \He} \left\{\inf_{\se \in \Se}~ G^{\Phi}_{N_{\Phi}}(\se) \right\} = \sup_{\Phi \in \He} \left\{\inf_{\se \in \Se} \left[ \sum_{i=0}^{N_{\Phi}(\se)-1} H_i(\se) \cdot \Delta_i X(S) \right]\right\} = 0.
\end{equation}
\end{definition}

In \cite{ferrando} it is shown that this property is also sufficient to obtain a pricing interval for financial derivatives. The next Proposition shows that $0$-neutrality is weaker than arbitrage-free.

\begin{proposition} \label{prop:0neutral}
Let $\Me=\Se \times \He$ be  an arbitrage-free trajectory based market. Then, $\Me$ is $0$-neutral.
\end{proposition}
\begin{proof}
We are going to prove it by contraposition. Note that if $\Me= \Se \times \He$ is a trajectory based market, $0 \in \He$, then it is always true that
\[  \sup_{\Phi \in \He}\{ \inf_{\se \in \Se}~ G^{\Phi}_{N_{\Phi}}(\se)\} \ge 0.\]
That is, if $\Me$ is not $0$-neutral, there exist a portfolio $\Phi$ such that
\[  \inf_{\se \in \Se} ~G^{\Phi}_{N_{\Phi}}(\se) > 0 .\]
Thus $G^{\Phi}_{N_{\Phi}}(\se) > 0$  for all $\se \in \Se$. Then
\[ V^{\Phi}_{N_{\Phi}}(\se) = V^{\Phi}_0 + G^{\Phi}_{N_{\Phi}}(\se) > V^{\Phi}_0\]
for all $\se \in \Se$. Then $\Phi$ is an arbitrage porfolio.
\end{proof}

It is clear how to generate simple examples of $0$-neutral markets which
contain arbitrage (see \cite{ferrando} and \cite{degano}).
Following \cite[Section 1.2]{pliska} it is possible to define another properties of the market, closely related to $0$-neutral and arbitarge-free, namely \emph{dominant portfolios} and \emph{law of one price}. Under the appropriate hypotheses, the following chain of implications for a trajectory based market holds:
\[ \boxed{ \mbox{Arbitrage-free} \Rightarrow \mbox{No dominant portfolios} \Rightarrow \mbox{$0$-neutral}\Rightarrow \mbox{Law of one price.}}\]

At this point, we have introduced enough properties of multidimensional trajectory markets in order to address our goal of characterizing no-arbitrage symmetry transformations.

\subsection{Relationships Between Local and Global Properties}\label{global&local}

From the definitions, it is not clear how to construct arbitrage-free or $0$-neutral markets. For the case of semi-bounded markets, one can obtain necessary and sufficient conditions, only involving local properties of the trajectory set, implying
trajectorial markets that are arbitrage-free (or $0$-neutral). Such characterizations play an analogous role to the equivalence of no arbitrage stochastic markets and the possibility to equivalently modify the stochastic process into a martingale process. In fact, in the arbitrage-free case the local trajectorial conditions correspond to a probability free notion of a martingale sequence (see \cite{ferrando2}). We will use these characterizations to pose and answer our opening question on no-arbitrage preserving transformations.

At the $ k$-th stage, the information about the future available to investors is that $ \se $ is an element of the set
\begin{equation} \nonumber
\mathcal{S}_{(\se,k)}\equiv \left\lbrace \se' \in \mathcal{S}: \se'_i= \se_i, 0 \le i \le k \mbox{ and } M(\se') > k \right\rbrace \subseteq \Se.
\end{equation}
We will call the pair $(\se,k)$ a \emph{node} and will refer to the set $\Se_{(\se,k)}$ as \emph{trajectory set conditioned at the node $(\se,k)$}. The future information contained in $ \tilde{\se} \in \Se_{(\se, k)} $ depends on the past only through ${\bf S}_0, \ldots, {\bf S}_k$.
The multiple number of trajectories  emanating from a node
reflects the non-deterministic nature of the assets' time evolution.
As trajectories unfold more coordinates become available and so the investor increases his knowledge about possible future scenarios. This is expressed mathematically as
\[ \Se_{(\se,k')} \subseteq \Se_{(\se,k)}\]
for $k'>k$. The following notation will also be used
\begin{equation}\label{eqn:conjdelta}
  \Delta X(\Se_{(\se,k)}) \equiv \{ \Delta_k X(S'):~ \se' \in \Se_{(\se,k)} \} \subseteq \mathbb{R}^{d},
\end{equation}
where $\Delta_k X(S')=X(S'_{k+1})-X(S'_k)$ has been introduced before.

\vspace{.1in}
We will refer as {\it local} to any property relative to a node $ (\se, k) $ and only involving elements of $\Delta X(\Se_{(\se,k)})$.

The definitions below, are the local counterpart of those of arbitrage-free and $0$-neutral for the whole market. We are then going to derive the global properties from the local ones.

\begin{definition}[Local notions] \label{localWithRespectToH}
  Given a trajectory based market $\Me= \Se \times \He$, let $\se \in \Se$ and $k \ge 0$.
\begin{enumerate}
  \item $(\se, k)$ is called an {\it arbitrage-free node {\it with respect to} $\mathcal{H}$} if
  \[
  [H_k(\se) \cdot \Delta_k X(S')=0 \quad \forall \se' \in \Se_{(\se,k)}]~ \mbox{or}~ [\inf_{\se' \in \Se_{(\se,k)}} H_k(\se) \cdot \Delta_k X(S') < 0],
  \]
  for all $H \in \He_{\Se}$ (the latter as in (\ref{displaySoDefinitionIsNotLost})).
 \item $(\se, k)$ is called a {\it $0$-neutral node  with respect to} $\mathcal{H}$ if, for all $H \in \He_{\Se}:$
  \[
  \inf_{\se' \in \Se_{(\se,k)}} H_k(\se) \cdot \Delta_k X(S') \le 0.
  \]
\end{enumerate}
  $\Me$ is called {\it locally arbitrage-free ($0$-neutral)} if each $(\se, k)$ is an arbitrage-free ($0$-neutral) node w.r.t.~$\He$. A node that is not arbitrage-free w.r.t.~$\He$, will be called an {\it arbitrage node} w.r.t.~$\He$.
\end{definition}

Notice that an arbitrage-free node w.r.t.~$\He$ is always $0$-neutral w.r.t.~$\He$. Clearly, there are natural examples of nodes which are $0$-neutral w.r.t.~$\He$ but no arbitrage-free w.r.t.~$\He$ (hence these are arbitrage nodes). It is then of interest to indicate that there are results (\cite{ferrando}) that justify option prices obtained for general $0$-neutral markets (in particular these markets may contain $0$-neutral nodes which are arbitrage nodes w.r.t.~$\He$).

Admittedly, attaching the qualifier ``w.r.t. $\mathcal{H}$" to some of the above notions does not play a substantial role in the paper. In fact,
Proposition \ref{suffCondForArb-freeAnd0Neutral} below provides sufficient conditions on trajectory nodes that imply that those nodes are arbitrage-free ($0$-neutral) w.r.t.~{\it any} (admissible) $\He$.

The conclusions in Proposition \ref{suffCondForArb-freeAnd0Neutral} below are consequences of characterizations given by Propositions \ref{localDefinitions} and \ref{localDefinitionsII} in Section \ref{local geometric section}.

\begin{proposition}\label{suffCondForArb-freeAnd0Neutral}
Given a trajectory set $\Se$, consider a
node $(\se,k)$.

\begin{enumerate}
\item If:
\begin{equation} \label{arbitrage-freeNode}
0 \in \mathrm{ri}\left(\mathrm{co}\left(\Delta X(\Se_{(\se,k)})\right)\right).
\end{equation}
then $(\se,k)$ is an arbitrage-free node w.r.t. any (admissible) $\He$.

\item If:
\begin{equation} \label{0-neutral node}
0 \in \mathrm{cl}\left(\mathrm{co}\left(\Delta X(\Se_{(\se,k)})\right)\right).
\end{equation}
then $(\se,k)$ is a $0$-neutral node w.r.t. any (admissible) $\He$.
\end{enumerate}
\end{proposition}
According with these results we introduce the following notions which will play a crucial role for the remaining of the paper.

\begin{definition}[$\mathcal{H}$-Independent local properties] \label{trSpaceLocArbFree}
A node $(\se,k)$ is called \emph{arbitrage-free} if (\ref{arbitrage-freeNode}) is satisfied; it is called \emph{$0$-neutral} if (\ref{0-neutral node}) is satisfied. We call $\mathcal{S}$ locally arbitrage-free (locally $0$-neutral), if every node $({\bf S}, k)$ is arbitrage-free ($0$-neutral).
\end{definition}
\begin{remark}
The above definitions rely on a numeraire (through the perspective function $X$). We will show in Section \ref{sec:invarianceOfLocalProperties} that once the properties hold for one numeraire, they hold for any numeraire.
\end{remark}
Therefore, if $\mathcal{S}$ is locally arbitrage-free (locally $0$-neutral), then
$\mathcal{M}= \mathcal{S} \times \mathcal{H}$ is locally arbitrage-free (locally $0$-neutral) for any (admissible) $\mathcal{H}$.

\begin{remark} Condition (\ref{arbitrage-freeNode}) appears in the stochastic literature as equivalent to one step arbitrage-free markets \cite[Lemma 3.42]{cutland}, \cite[Prop 3.3.4]{elliot}, \cite[Cor 1.50]{follmer} , \cite{jacod}.
\end{remark}

\vspace{.1in}
The local notions in Definition \ref{trSpaceLocArbFree}   allow us to ensure global conditions on a trajectory based market. In particular, the results in the rest of this section characterize an arbitrage-free market ($0$-neutral) by means of arbitrage-free ($0$-neutral) nodes w.r.t.~$\He$.

\begin{theorem}[No arbitrage: local implies global] \label{teo:locarbfree}
If $\mathcal{M} = \mathcal{S}\times \mathcal{H}$ is locally arbitrage-free (as per Definition \ref{localWithRespectToH}) and  semi-bounded, then $\Me$ is arbitrage-free (as per Definition \ref{ArbitrageDefinition}).
\end{theorem}

\noindent
See proof in Appendix \ref{proofsAndResults}.

\vspace{.1in}
In order to establish a converse to Theorem \ref{teo:locarbfree}, consider $\xi \in \mathbb{R}^{d}$, $\se \in \Se$ and $k \ge 0$, let us define the function $\xi_i^{(\se,k)}: \Se \rightarrow \mathbb{R}^{d}$, for any $i\ge 0$, by
\[
\xi_i^{(\se,k)}(\se')= \left\lbrace \begin{array}{ll}
\xi & \mbox{if }\; i=k \;\mbox{and}\; \se' \in \Se_{(\se,k)}.\\
0 & \mbox{otherwise.}
\end{array} \right.
\]
Given $V_0$, we can obtain from the equations \eqref{eqn:defH1} and \eqref{eqn:defH2} a sequence of functions $\{ \xi^0_i \}_{i \ge 0}$ in such a way that the sequence
\begin{equation}\label{restrictedPort}
\Xi^{(\se,k)}=\{(\xi_i^0,\xi_i^{(\se,k)})\}_{i\ge 0}
\end{equation}
be self-financing. Also, defining $N_{\Xi^{(\se,k)}}(\se')=k+1$ for all $\se' \in \Se$, it is easy to see that $\Xi^{(\se,k)}$ is a portfolio. We will call this type of portfolios as \emph{restricted portfolios} at the node $(\se,k)$.

\begin{proposition}[No arbitrage: global implies Local] \label{globalImpliesLocalNA}
If $\mathcal{M} = \mathcal{S} \times \mathcal{H}$ is arbitrage free and  the restricted portfolios belong to $\He$ then, $\mathcal{S}$ is locally arbitrage-free (as per Definition \ref{trSpaceLocArbFree}). In particular $\mathcal{M}$ is locally arbitrage-free.
\end{proposition}
\noindent
See proof in Appendix \ref{proofsAndResults}.

\vspace{.1in}
We now carry out a similar analysis for the notion of $0$-neutral. The following Theorem shows that a trajectory based market will be $0$-neutral if it is locally $0$-neutral.

\begin{theorem}[$0$-neutral: local implies global] \label{teo:0neutral}
  Let $\Me= \Se \times \He$ be a semi-bounded  trajectory market. Then if $\Me$ is locally $0$-neutral (as per Definition \ref{localWithRespectToH}) then, $\Me$ is $0$-neutral (as per Definition \ref{0Neutral}).
\end{theorem}
\noindent
See proof in Appendix \ref{proofsAndResults}.

\noindent

\begin{proposition}[$0$-neutral: global implies local] \label{0NeutralGlobalImpliesLocal}
Let $\Me= \Se \times \He$ be a $0$-neutral trajectory market such that the restricted portfolios belong to $\He$. Then, any node $ (\se, k) $ is a $0$-neutral node (in particular, $(\se,k)$ is $0$-neutral with respect to $\mathcal{H}$).
\end{proposition}
\noindent
See proof in Appendix \ref{proofsAndResults}.

\section{Geometric Characterizations} \label{geometricCharacterizations}

We develop geometric characterizations for the local notions introduced in the previous section. Definition \ref{localGeometry2} below is a stronger version of
Definition \ref{localWithRespectToH} that dispenses of the qualifier ``w.r.t. $\mathcal{H}$" present in the latter definition.

\subsection{Local Geometric Characterizations}\label{local geometric section}
\begin{definition}[Disperse and $0$-neutral sets]  \label{localGeometry2}
Consider a set $E \subseteq \mathbb{R}^d$; $E$ is called {\it disperse} if for each  $h \in \mathbb{R}^d$:
\begin{equation} \label{geometricNoArbitrage}
    [h \cdot y= 0~\forall ~y \in E]~\mbox{or}~~[\inf_{y \in E}  h \cdot y < 0  ~~\&~~~\sup_{y \in E}  h \cdot y > 0].
\end{equation}
$E$ is called $0$-neutral if for each  $h \in \mathbb{R}^d$:
\begin{equation}  \label{geometric0Neutral}
    [\inf_{y \in E}  h \cdot y \leq 0 ~\&~\sup_{y \in E}  h \cdot y \geq 0].
\end{equation}
\end{definition}
Notice that (\ref{geometric0Neutral}) is equivalent to just requiring the validity of one of the two inequalities appearing in the conjunction in (\ref{geometric0Neutral}). Similarly,~(\ref{geometricNoArbitrage}) is equivalent to  $[h \cdot y= 0~\forall ~y \in E]~\mbox{or}~~[\inf_{y \in E}  h \cdot y < 0]$
(this later inequality could be replaced by $[\sup_{y \in E}  h \cdot y > 0]$). We have written  Definition \ref{localGeometry2} in its present form for emphasis.

Results and notions from convex analysis that we will rely upon in this section are detailed in Appendix \ref{app:conv}.

\begin{proposition} \label{localDefinitions}
Let $E \subseteq \mathbb{R}^d$.
\begin{equation} \nonumber
 E ~\mbox{is disperse iff~~}~~~~ 0 \in \mathrm{ri}(\mathrm{co}(E)).
 \end{equation}
\end{proposition}
\begin{proof}
Assume first that $E$ is disperse. In order to proceed to deduce a contradiction, we assume that $0 \notin \mathrm{ri}(\mathrm{co}(E))$; by the separation Theorem \ref{teo:separacion}, there exists $\xi \in \mathbb{R}^d$ such that:
\begin{itemize}
\item $\xi \cdot x \ge 0$ for all $x \in \mathrm{ri}(\mathrm{co}(E))$, and
\item $\xi \cdot x^{\ast} > 0$ for some $x^{\ast} \in \mathrm{ri}(\mathrm{co}(E)).$
\end{itemize}
Then, by means of Proposition
\ref{prop:ri2},  it follows that 
for all $x \in \mathrm{ri}(\mathrm{co}(E))$
\[ (\alpha x + (1-\alpha)~y) \in \mathrm{ri}(\mathrm{co}(E)~\mbox{for all}~y \in E~\mbox{and}~\alpha \in (0,1].\]
Therefore,
\[ \xi \cdot (\alpha x + (1-\alpha)~y) = \alpha \xi \cdot x + (1-\alpha)\xi \cdot y  \ge 0~~\mbox{for all}~y \in E~\mbox{and}~~ \alpha \in (0,1].\]
 It then follows that $ \xi \cdot y  \ge 0 $ for all $y \in E$. This contradicts the fact that $E$ is disperse.

\vspace{.1in}

Conversely, assume that $ 0 \in \mathrm{ri}(\mathrm{co}(E))$. We may assume there exists $\hat{y} \in E$
such that $h \cdot \hat{y} \neq 0$. It is enough to establish that
$\inf_{y \in E} h \cdot y <0$, we may then
assume there exists $y^{\ast} \in E$ such that $h \cdot y^* > 0$. As $y^* \in \mathrm{co}(E)$ and $0 \in \mathrm{ri}(\mathrm{co}(E))$, it follows from Proposition \ref{prop:ri} in Appendix \ref{app:conv} that there exists $\epsilon >0$ such that
\[ -\epsilon y^* \in \mathrm{co}(E).\]
Then, it follows from Theorem \ref{carath} that there exists $y^{(1)}, \dots, y^{(d+1)} \in E$ such that
\[ -\epsilon y^*= \lambda_1 y^{(1)}+ \dots + \lambda_{d+1} y^{(d+1)}  \mbox{ with } \sum_{i=1}^{d+1} \lambda_i = 1, \quad \lambda_i \ge 0.\]
Then
\[ 0 > -\epsilon \left(h \cdot y^* \right)= \sum_{i=1}^{d+1}\lambda_i \left(h \cdot y^{(i)}\right).\]
Therefore, there must be some $1 \le j \le d+1$ such that $h \cdot y^{(j)} < 0$, and then
\[ \inf_{y \in E} h \cdot y < 0.\]
\end{proof}

\noindent Similarly to Proposition \ref{localDefinitions}, the following result characterizes the $0$-neutral property of $E$.
\begin{proposition} \label{localDefinitionsII}
Let $E \subseteq \mathbb{R}^d$.
\begin{equation} \nonumber
E~\mbox{is}~~ 0-\mbox{neutral iff~~}~~ 0 \in \mathrm{cl}(\mathrm{co}(E)).
\end{equation}
\end{proposition}
\begin{proof}
Assume $E$ is $0$-neutral and $0 \notin \mathrm{cl}(\mathrm{co}(E))$. Since the closure of a convex set is a convex set (Proposition \ref{prop:clconv}) and closed, by Theorem \ref{teo:separacion} in Appendix \ref{app:conv}, it follows that there exists  $\xi \in \mathbb{R}^d$ such that $\inf \xi \cdot y > 0$ where the infimum is over all $y \in \mathrm{cl}(\mathrm{co}(E))$. Thus
\[ \inf_{y \in E} \xi \cdot y >0,\] which contradicts our hypothesis.

Assume now $0 \in \mathrm{cl}(\mathrm{co}(E))$.
It is enough to show that $\inf_{y \in E} h \cdot y \leq 0$ for any $h \in \mathbb{R}^d$. To proceed by contradiction, assume there exists $h \in \mathbb{R}^d$, and $\epsilon >0$ such that $\epsilon < \inf_{y \in E} h \cdot y \le h \cdot y$ for all $y \in E$ (otherwise we are done).
  From our hypothesis, there exists a sequence $\{ x_j\}_{j = 1}^{\infty} \subseteq \mathrm{co}(D)$ such that $x_j \rightarrow 0$ as $j \rightarrow \infty$. By Theorem \ref{carath} in Appendix \ref{app:conv}, for each $x_j$ there exists $y^{(1,j)}, \dots ,y^{(d+1,j)} \in E$ such that
  \[ x_j= \lambda^j_1 y^{(1,j)}+ \dots + \lambda^j_{d+1} y^{(d+1,j)}\; \mbox{with}\; \sum_{i=1}^{d+1} \lambda^j_i=1, \quad \lambda^j_i \ge 0.\]
  Then
  \begin{eqnarray*}
  0&=& h \cdot 0 = h \cdot \left( \lim_{j \rightarrow \infty} x_j \right) =  \lim_{j \rightarrow \infty} \left(h \cdot  x_j\right)=\\ &=& \lim_{j \rightarrow \infty} \sum_{i=1}^{d+1} \lambda^j_i \left(h \cdot y^{(i,j)} \right)\ge \lim_{j \rightarrow \infty} \sum_{i=1}^{d+1} \lambda^j_i \epsilon = \epsilon.
  \end{eqnarray*}
   which is a contradiction,  thus, we conclude.
\end{proof}

Lemma \ref{translation} below uses the following notation,  for $E\subset \mathbb{R}^d$, and $x_0\in \mathbb{R}^d$, let
\[E-x_0 \equiv \{x-x_0 : x\in E\}\subset \mathbb{R}^d.\]

Since the translation $t_{x_0}:\mathbb{R}^d \rightarrow \mathbb{R}^{d}$, given by $t_{x_0}(x)=x-x_0$, is of the form (\ref{ratioFunction}) in Appendix \ref{app:conv}, and an homeomorphism, the lemma below follows.
\begin{lemma}\label{translation}
\[
0 \in \mathrm{ri}\left(\mathrm{co}(E-x_0)\right)\quad \mbox{iff} \quad x_0\in\mathrm{ri}\left(\mathrm{co}(E)\right).
\]
\[
0 \in \mathrm{cl}\left(\mathrm{co}(E-x_0)\right)\quad \mbox{iff} \quad x_0\in\mathrm{cl}\left(\mathrm{co}(E)\right).
\]
\end{lemma}

\subsection{Convexity Preserving Maps}\label{sec:geometryPreserving}

In order to identify transformations that preserve no-arbitrage ($0$-neutrality) and in view of Proposition \ref{localDefinitions} (Proposition \ref{localDefinitionsII}) and Lemma \ref{translation} we first look for transformations $F:\mathbb{R}^d \rightarrow \mathbb{R}^{d'}$  preserving relative interiors or closures of convex sets in $\mathbb{R}^d$.

The notions introduced below are
expanded in Appendix \ref{app:conv} where we also provide due references and
introduce related definitions and further results.

\begin{definition}[Strict inversely convexity preserving]
Let $\mathbb{V}$ and $\mathbb{V}'$ be real linear spaces, and $C\subset \mathbb{V}$ a nonempty convex subset.
A map $g:C \rightarrow \mathbb{V}'$ is called \emph{strict inversely convexity preserving} if
\begin{equation}\label{invConvexPres}
g((x,y)) \subseteq (g(x),g(y)) \quad \mbox{for all}\; x,y \in C,
\end{equation}
where $(x,y)=\{tx + (1-t)y:0<t<1\}$ (with a similar definition for $[x,y]$, see Appendix \ref{app:conv}). Moreover, $g$ is said to {\it preserve segments} strictly if equality holds in (\ref{invConvexPres}).
\end{definition}

The next lemma provides necessary and sufficient conditions on a transformation $F$  in order to preserve $\mathrm{ri}(\mathrm{co}(E))$, for $E\subset \mathbb{R}^d$.

\begin{lemma} \label{teo:preservaarb}
 Let $C \subseteq \mathbb{R}^d$ be a convex set, $x\in C$ and $E\subset C$.

\noindent $F:C  \rightarrow \mathbb{R}^{d'}$ is a strict inversely convexity preserving map if and only if
\[
x \in \mathrm{ri}\left(\mathrm{co}(E)\right)~~\mbox{implies}~~~F(x) \in \mathrm{ri}\left(\mathrm{co}(F(E))\right). 
\]
\end{lemma}

\begin{proof} Let $x \in \mathrm{ri}\left(\mathrm{co}(E)\right)$ and $b'\in\mathrm{co}(F(E))$. Assume first that $F(C)$ is contained in a straight line. Fix $b\in\mathrm{co}(E)$, from Corollary \ref{cor:ri}, there exists $a\in\mathrm{co}(E)$ such that $x\in (a,b)$. Then by hypothesis on $F$ and Proposition \ref{prop:preservaconvexos}, $F(a), F(b)\in F(\mathrm{co}(E))\subset \mathrm{co}(F(E))$, and
\[
F(x)\in F((a,b))\subset(F(a),F(b)).
\]
Now, since $\mathrm{co}(F(E))$ is a segment, because it is contained in a straight line, it follows that:

If $b'\in (F(x),F(b))$, or $F(b)\in (F(x),b')$ then
\[
F(x)\in (F(a),b')\subset \mathrm{co}(F(E)).
\]
On the other hand $F(x)\in (F(b),b')\subset \mathrm{co}(F(E))$. Thus in any case by Corollary \ref{cor:ri}
\[
F(x)\in \mathrm{ri}\left(\mathrm{co}(F(E))\right).
\]

If $F(C)$ is not contained in a straight line, by Theorem \ref{teo:pales}, $F$ preserves segments strictly, then $\mathrm{co}(F(E))= F(\mathrm{co}(E))$, so $b'=F(b)$ with $b\in\mathrm{co}(E)$.

As before, there exists $a\in\mathrm{co}(E)$ such that $x\in (a,b)$. Then \\ $F(a)\in \mathrm{co}(F(E))$ and
\[
F(x)\in F((a,b))=(F(a),F(b)),
\]
which also leads to $F(x)\in \mathrm{ri}\left(\mathrm{co}(F(E))\right)$.

To establish the converse, consider the case $E= \{a,b\}$ then \[(a,b)=\mathrm{ri}\left(\mathrm{co}(E)\right)~~\mbox{and}~~(F(a), F(b))= \mathrm{ri}\left(\mathrm{co}(F(E))\right).\]
Now from our hypothesis,
\begin{equation} \label{holdsForAllEAndAllx0}
F(x) \in  (F(a), F(b)),~~\mbox{for any}~~ x \in (a,b),
\end{equation}
therefore, $F$ is strict inversely convexity preserving.
\end{proof}

Observe that, from Lemma \ref{teo:preservaarb}, Proposition \ref{localDefinitions} and Lemma \ref{translation}, $F$ is strict inversely convexity preserving if it preserves disperse sets.

In a similar way than before, $F$ preserves $\mathrm{cl}(\mathrm{co}(E))$, if for $x\in [a,b]=\mathrm{cl}(\mathrm{co}(\{a,b\}))$ it holds that
\[
F(x)\in \mathrm{cl}\left(\mathrm{co}(\{F(a),F(b)\})\right)= [F(a),F(b)].
\]
That is, $F$ need to be inversely convexity preserving (see Proposition \ref{prop:preservaconvexos} in Appendix \ref{app:conv}). However this condition on its own is not sufficient (see next Lemma and Example \ref{counterexample2}.)
\begin{lemma} \label{lem:preserva0neutral}
Let $C \subseteq \mathbb{R}^d$ a convex set, $E\subset C$ and $F:C  \rightarrow \mathbb{R}^{d'}$ a continuous inversely convexity preserving map. If $x_0\in\mathrm{cl}(\mathrm{co}(E))$, then
\[F(x_0)\in \mathrm{cl}\left(\mathrm{co}\left(F(E)\right)\right).\]
\end{lemma}
\begin{proof}
 By continuity of $F$, for all $E \subseteq C$ holds $F(\mathrm{cl}(E)) \subseteq \mathrm{cl}(F(E))$. Furthermore, since $F$ is a inversely convexity preserving map,
\[
F(\mathrm{co}(E))\subset \mathrm{co}(F(E)).
\]
Thus, since $x_0 \in \mathrm{cl}(\mathrm{co}(E))$,
\[
F(x_0) \in F(\mathrm{cl}(\mathrm{co}(F(E))) \subseteq \mathrm{cl}(F(\mathrm{co}(E)) \subseteq \mathrm{cl}\left(\mathrm{co}\left(F(E)\right)\right).
\]
\end{proof}
\begin{ej}\label{counterexample2}
The hypothesis of continuity in Lemma \ref{lem:preserva0neutral} can not be removed, to see this consider $F:\mathbb{R}\rightarrow \mathbb{R}$ given by
\[
F(x)=\left\{\begin{array}{ccc}
            x & \mbox{if} & x\le 0 \\
            x+1 & \mbox{if} & x>0,
\end{array}\right .
\]
is inversely convexity preserving, but not continuous and $0\in \mathrm{cl}((0,1])$, but $$F(0)=0\notin [1,2]=\mathrm{cl}(\mathrm{co}(F((0,1]))).$$
\end{ej}

\subsection{Induced Transformations}

 As indicated in Section \ref{currencyAndNumeraire}, we have taken a standard view in which the original
 sequency $S_i$ is given in a currency numeraire and then the sequence $X(S_i)$ is given in another (arbitrary) numeraire. Since we look for transformations between trajectories of financial markets that preserve their local properties
we will be dealing with two associated functions, $f$ and $F$ the former acting on $S_i$ and the latter on $X(S_i)$. So we will have $\mathbb{R}^{d+1} \xrightarrow{f}\mathbb{R}^{d'+1}$ and $\mathbb{R}^{d} \xrightarrow{F}\mathbb{R}^{d'}$. One could proceed differently and develop an approach which abstracts away this multiplicity; nonetheless, we have decided to proceed the way we do as in practice that is how data is usually presented. This decision  makes our results more readily applicable albeit at the price of some complications.

 Let $X$ and $X'$ be the perspective functions over $\mathbb{R}^{d+1}$ and $\mathbb{R}^{d'+1}$ respectively, as defined in \eqref{eqn:X}. Since local properties are based on properties of discounted values, the function $f$ should induce a function $F:\mathbb{R}^{d}\rightarrow \mathbb{R}^{d'}$, in such a way that the following diagram commutes,
\begin{equation}\label{diagram}
\xymatrix{
\mathrm{dom}~X \ar[r]^f \ar[d]_X & \mathrm{dom}~X' \ar[d]^{X'} \\
\mathrm{Im}~X\subset \mathbb{R}^d \ar@{.>}[r]_{F} & \mathrm{Im}~X'\subset\mathbb{R}^{d'} }
\end{equation}
that is, for $s\in\mathrm{dom}~X$, $F(X(s))\equiv X'(f(s))$. Therefore, if $X(s)=X(\tilde{s})$ for $s,\tilde{s}\in\mathrm{dom}~X$ we will require that $X'(f(s))=X'(f(\tilde{s}))$, which gives a condition on $f$ as we describe next.

Assume $s =(s^0,\ldots,s^d),\; \tilde{s} =(\tilde{s}^0,\ldots,\tilde{s}^d)$, then
\begin{eqnarray*}
\left(\frac{s^1}{s^0},\dots,\frac{s^d}{s^0}\right)=X(s)=X(\tilde{s})=\left(\frac{\tilde{s}^1}{\tilde{s}^0},\dots,\frac{\tilde{s}^d}{\tilde{s}^0}\right) &\Leftrightarrow & s=\frac{s^0}{\tilde{s}^0}\tilde{s},
\end{eqnarray*}
from where $f$ needs to satisfy
\begin{equation}\label{scalar}
f(\lambda s) = \mu_{\lambda,s} ~f(s),\;\lambda,\,\mu_{\lambda,s} >0,\quad \mbox{for any}\;s\in\mathbb{R}^{d+1}.
\end{equation}
We will require (\ref{scalar}) in our main result Theorem \ref{teo:preservaarb2}
(as well as in Theorem \ref{teo:preserva0neutral}).

Lemma \ref{teo:preservaarb} shows that a map $F$ preserving no-arbitrage  is necessarily strict inversely convexity preserving. Lemma \ref{lem:diagconm}, item 1, below provides sufficient conditions on $f$ to establish the said property of $F$; on the other hand, Example \ref{counterexample1} shows that the assumptions on $f$, while being sufficient, are not necessary.
\begin{lemma} \label{lem:diagconm}
Let $f: \mathrm{dom}~X \rightarrow \mathrm{dom}~X'$ be a function satisfying (\ref{scalar}). Then, there exists a unique map $F:\mathrm{Im}~X \to \mathrm{Im}~X'$ which makes commutative the diagram (\ref{diagram}). Moreover
\begin{enumerate}
\item If $f$ is (strict) inversely convexity preserving then $F$ is (strict) inversely convexity preserving.
\item $F$ is continuous iff $f$ is continuous.
\end{enumerate}
\end{lemma}
\begin{proof}
For all $x \in \mathrm{Im}~X$ there exist $s \in \mathrm{dom}~X$ such that $X(s)=x$. The only way to define $F$ is then $F(x)=X'(f(s))$ for all $x \in \mathbb{R}^d$ and it is well defined by condition (\ref{scalar}).

Let's see that $F$ is a strict inversely convexity preserving map if $f$ is
assumed to satisfy that property. Fix $\hat{x},\tilde{x} \in \mathrm{Im}~X$ and let $x \in \mathrm{Im}~X$ such that
 \[
 x= \alpha \hat{x} + (1-\alpha) \tilde{x} \mbox{ with } 0<\alpha<1.
 \]
Then, there exists $\hat{s},\tilde{s} \in \mathrm{dom}~X$ such that $X(\hat{s})=\hat{x}$ and $X(\tilde{s})=\tilde{x}$. Moreover, since $X$ is a strict segment preserving map (Theorem \ref{prop:afin}), there exists $\beta\in(0,1)$ such that
\[x=X(\beta \hat{s}+(1-\beta)\tilde{s}).\]
Then since $X'$ is strict segment preserving,
\[ F(x) = X'\left(f\left(\beta \hat{s}+(1-\beta)\tilde{s}\right)\right)\in \left(X'(f(\hat{s})),X'(f(\tilde{s}))\right)= \left(F(\hat{x}),F(\tilde{x})\right).
\]
The proof for the case inversely convexity preserving, is similar. This gives item {\it 1}.

For item {\it 2.}, observe that the perspective functions $X, X'$ are continuous and open. The last assertion follows because if $Q$ is an open cube in $\mathbb{R}^d$, and $a^0<b^0$ positive real numbers, then
\[
X((a^0,b^0)\times Q)=\bigcup_{r\in (a^0,b^0)}\frac 1r \, Q,
\]
is open in $\mathbb{R}^d$. Thus, by composition $F$ is continuous iff $f$ is continuous.
\end{proof}
\begin{ej}\label{counterexample1} The converse in item {\it 1} of Lemma \ref{lem:diagconm}  is not valid. Consider  $f: \{(x,y,z)\in \mathbb{R}^3: x>0\}\rightarrow \{(x,y,z)\in \mathbb{R}^3: x>0\}$ given by
\[
f(x,y,z)= (\frac{x}{x^2+y^2},\frac{y}{x^2+y^2},\frac{z}{x^2+y^2}).
\]
The induced function $F$ is the identity function on $\mathbb{R}^2$ but $f$ is not inversely convexity preserving because $(1,0,0)\in\left((1,-1,0),(1,1,0)\right)$, but
\[
f(1,0,0) = (1,0,0) \notin \left(f(1,-1,0),f(1,1,0)\right)=\left((1/2,-1/2,0),(1/2,1/2,0)\right).
\]
\end{ej}

Under the additional hypothesis that $\mbox{Im}~ F$ contains a nondegenerate triangle, the next Theorem characterizes those $f: \mathrm{dom}~X \rightarrow \mathrm{dom}~X'$ inducing a strict inversely convexity preserving map $F$. That hypothesis  on $\mbox{Im}~ F$ is equivalent to $\mbox{Im}~ F$ not being contained in a straight line. As a complement, Lemma \ref{triangle} shows that this last condition on $\mbox{Im}~ F$ holds iff $\mathrm{Im}~ f$ is not contained in a $2$ dimensional subspace.

\noindent
The appearance of $``0"$ in $\frac{f^0(s)}{L^0(s)}$ below merely reflects our arbitrary choice of $S^0$ as numeraire,
choosing $S^k$ as numeraire will result in the appearance of $\frac{f^k(s)}{L^k(s)}$ in the next result (see Section \ref{example} for an example).
\begin{theorem}\label{functionInducingISCP}
Assume $f: \mathrm{dom}~X \rightarrow \mathrm{dom}~X'$ satisfying (\ref{scalar}), induces a strict inversely convexity preserving map $F$,  such that $\mathrm{Im} F$ is not contained in a straight line. Then
\begin{equation}\label{fractionalf}
f(s)= \frac{f^0(s)}{L^0(s)}L(s),
\end{equation}
with $f^0$ (its first coordinate function) satisfying (\ref{scalar}), $L:\mathbb{R}^{d+1}\rightarrow \mathbb{R}^{d'+1}$, a linear map, with and $L^0, f^0> 0$ on $s^0 > 0$, ($L^0$ the first coordinate of $L$).

\noindent Conversely, if $f$ has the form (\ref{fractionalf}) and satisfies the properties listed after that formula display, then it induces a strict inversely convexity preserving map $F$.
\end{theorem}
\begin{proof} Let us consider first that $\mathrm{dom}~X =\{s\in \mathbb{R}^{d+1}:s^i>0\;\forall i\}$, so $\mathrm{dom}~F =\{x\in \mathbb{R}^{d}:x^i>0\;\forall i\}$ is convex, then by Theorem \ref{teo:pales}
\begin{equation}\label{fractionalF}
X'(f(s)) = F(X(s)) =\frac{(A^1(X(s))+b^1,...,A^{d'}(X(s))+b^{d'})}{B(X(s))+c}.
\end{equation}
It then follows that for $1\le i \le d'$,
\[
\frac{f^i(s)}{f^0(s)}=F^i(X(s))= \frac {A^i(X(s))+b^i}{B(X(s))+c}=
\frac{a_{i,1}\frac{s^1}{s^0}+ \cdots + a_{i,d}\frac{s^d}{s^0}+b^i}{B^1\frac{s^1}{s^0}+\cdots + B^d\frac{s^d}{s^0}+c}.
\]
Which can be written as $f^i(s) = f^0(s)\frac{L^i(s)}{L^0(s)}$, with
\[
L^i(s)= b^i s^0 + a_{i,1} s^1+ \cdots + a_{i,d}s^d,\;\mbox{and}\;L^0(s)= c s^0 + B^1 s^1 + \cdots + B^d s^d.
\]
From where, defining $L(s)=(L^0(s),L^1(s),\cdots,L^{d'}(s))$, (\ref{fractionalf}) holds with the expected conditions, since $f^0$ satisfies (\ref{scalar}) because $f$ do, and both $f^0, L^0 > 0$ on $s^0>0$.

Assume now that $\mathrm{dom}~X =\{s\in \mathbb{R}^{d+1}:s^0>0\}$, which implies that $\mathrm{dom}~F = \mathbb{R}^d$, then by \cite[Cor 1]{pales}, (\ref{fractionalF}) can be written with $B(x)+c\equiv 1$. Consequently (\ref{fractionalf}) holds with $L^0(s)=s^0$.

Conversely, if $f$ has the form (\ref{fractionalf}) with the required conditions, then satisfies (\ref{scalar}), because $f_0, L, L^0$ satisfy (\ref{scalar}), by hypothesis and linearity respectively, consequently by Lemma \ref{lem:diagconm} there exists $F$ such that $F(X(s))=X'(f(s))$. Let's show that $F$ is strict inversely convexity preserving.
\[
X'(f(s))= \left(\frac{L^1(s)}{L^0(s)},\cdots,\frac{L^{d'}(s)}{L^0(s)}\right),
\]
where for $1\le i\le d'$
\[
\frac{L^i(s)}{L^0(s)}= \frac{a_{i,0}s^0 +\cdots +a_{i,d}s^d}
     {a_{0,0}s^0+\cdots +a_{0,d}s^d}= \frac{a_{i,0}+a_{i,1}\frac{s^1}{s^0} +\cdots a_{i,d}\frac{s^d}{s^0}}
     {a_{0,0} +a_{0,1}\frac{s^1}{s^0}\cdots a_{0,d}\frac{s^d}{s^0}}=
     \frac{a_{i,0}+A^i(X(s))}{a_{0,0}+B(X(s))}.
\]
With $A^i(x)= a_{i,1}x^1 +\cdots +a_{i,d}x^d$, and $B(x)=a_{0,1}x^1+\cdots +a_{0,d}x^d$ in the last expression. Defining $A=(A^1,\cdots,A^{d'})$ and $b=(a_{1,0},\cdots,a_{d',0})$, it follows that
\begin{equation} \label{palesFormDerived}
F(x)= \frac{b+A(x)}{a_{0,0}+B(x)},
\end{equation}
Which is strict inversely convexity preserving by Theorem \ref{prop:afin}.
\end{proof}

\vspace{.1in}
\begin{lemma}
\label{triangle} Assume $f: \mathrm{dom}~X \rightarrow \mathrm{dom}~X'$ is a function satisfying (\ref{scalar}) and $F$ the induced function as in Lemma \ref{lem:diagconm}. Then, $Im\,F$ is contained in a straight line iff $Im\,f$ is contained in a 2-dimensional subspace.
\end{lemma}
\begin{proof}Assume that $f(s)=(y^0,\ldots,y^{d'})$ then
\[
F(X(s))=X'(f(s)) = \frac{1}{y^0}(y^1,\ldots,y^{d'}).
\]
It follows that
\[
Im\,F = \{z\in \mathbb{R}^{d'}:(1,z)\in \lambda\,(Im\,f),\;\mbox{for some}\;\lambda>0\}.
\]
If $Im\,f \subset \pi$, a 2-dimensional subspace, then
\[
Im\,F \subset \{z\in \mathbb{R}^{d'}:(1,z)\in \pi\},
\]
and this set is contained in the straight line $\pi\cap\{y^0=1\}\subset \mathbb{R}^{d'+1}$.

Conversely, assume there exist $s^1,s^2,s^3\in \mathrm{dom}~X$ such that $f(s^1), f(s^2), f(s^3)$ are l.i. Since $Im\,F$ is contained in a straight line, it follows that there exists $\alpha\in \mathbb{R}$ such that
\[
F(X(s^3)) = \alpha F(X(s^1)) + (1-\alpha)F(X(s^2)) =  \alpha X'(f(s^1)) + (1-\alpha)X'(f(s^2)),
\]
which leads to the contradiction
\[
f(s^3) = \frac{f^0(s^3)}{f^0(s^1)}\alpha f(s^1) + \frac{f^0(s^3)}{f^0(s^2)}(1-\alpha)f(s^2).
\]
\end{proof}

\section{No Arbitrage Invariance}
\label{sec:invarianceOfLocalProperties}

This section studies a class of transformations that do not change a given node's local properties of being arbitrage-free (this latter notion as per Definition \ref{trSpaceLocArbFree}). We also provide an explicit characterization for such symmetry transformations, this is achieved under a general, and weak, condition restricting their ranges.

As a special case, we will prove that the no-arbitrage property is unchanged under a  change of num\'eraire.  We also describe the similar results that apply for the property of $0$-neutral and, therefore, need also pursue some developments that apply to this concept as well. In general, the class of transformations studied should represent symmetries obeyed by any type of functional relationship among asset's prices resulting from no arbitrage considerations. In particular, if prices
$S$ satisfy a $h(S)=0$ relation, one then expects $h(S')=0$
where $S$ and $S'$ are related by a no-arbitrage symmetry as per Definition \ref{marketTransformation} below.
This fact is illustrated with an example  in Section \ref{example}.

Let $\Me= \Se \times \He$ and $\Me'= \Se' \times \He'$ be trajectory based markets, with $d+1$ assets, and $d'+1$ assets respectively. A transformation of $\Me$ onto $\Me'$, will be given by a function $f:\mathbb{R}^{d+1} \rightarrow \mathbb{R}^{d'+1}$ which will be called a {\it market transformation}. That is, to a trajectory $\se=(S,W,m)\in\Me$ corresponds a trajectory $\se'=(S',W',m')\in\Me'$, where $S'_k=f(S_k),\;k\ge 0$, and $W',m'$ are transformed in consequence. For instance, if $W$ represents the quadratic variation of the logarithm of the assets prices, then
\[W'_k = \sum_{i=0}^{k-1}(\log f(S_{i+1}) - \log f(S_i))^2.\] This example illustrates a case  when $W'$ can  be obtained from ${\bf S}'$. In other cases, when this is not possible,  $W'$ and $m'$ should be prescribed but, how this is actually done does not affect the developments in the present section.

\begin{definition} \label{marketTransformation}
A market transformation $f$, as above, which leaves invariant the arbitrage-free property ($0$-neutral property), as per Definition \ref{trSpaceLocArbFree},  of a given market's node will be called a
{\it no-arbitrage symmetry} ({\it $0$-neutral symmetry}).
\end{definition}
Therefore, if the node $({\bf S}, k)$ is arbitrage-free so will be
$({\bf S}', k)$ if $f$ is a no-arbitrage symmetry (similarly for a $0$-neutral symmetry). This remark also shows that the composition of no-arbitrage symmetries ($0$-neutral symmetries) is a no-arbitrage symmetry ($0$-neutral symmetry). We may refer to either type of symmetry as NAS (No-Arbitrage Symmetries) when there is no need to be specific.

\begin{remark}
The above notions depend on a choice of numeraire through Definition \ref{trSpaceLocArbFree} but we will prove in Corollary \ref{cor:numeraire} that a symmetry transformation
remains as such under a numeraire change. Of course the interest is in
general symmetry transformations $f:\mathbb{R}^{d+1} \rightarrow \mathbb{R}^{d'+1}$ that behave so for any possible node in any possible trajectory market (with corresponding dimension $d$) and that is the type of characterization we pursue.
\end{remark}

Recall from Definition \ref{localWithRespectToH} that local conditions are based on properties of the increment set $\Delta X(\Se_{(\se, k)})$, where $(\se, k)$ is a node of the market model. This set is totally determined by the values taken by the trajectories in the stage $k + 1$ and the value of $S_k$. To make this fact explicit, for each node $(\se, k)$ we introduce a notation for the set of reachable prices:
\begin{eqnarray} \label{eqn:conjsigma}
\Sigma_k(\se) &\equiv& \{ \hat{S}_{k+1}: \hat{\se}=(\hat{S},\hat{W},\hat{m}) \in \Se_{(\se,k)}\} \subseteq \mathbb{R}^{d+1}.
\end{eqnarray}

The next proposition (which follows from Lemma \ref{translation} in Section \ref{sec:geometryPreserving}) shows that local conditions can be rewritten in terms of the set $\Sigma_k(\se)$.

\begin{proposition} \label{prop:equivlocal}
Given a trajectory based set $\Se$, $\se=\{(S_i,W_i,m)\}_{i \geq 0} \in \Se$ and an integer $k \ge 0$.
\begin{enumerate}
\item The node $(\se,k)$ is arbitrage-free if, and only if,
\[ X(S_k) \in \mathrm{ri}\left(\mathrm{co}\left(X (\Sigma_k(\se))\right)\right).\]
\item The node $(\se,k)$ is $0$-neutral if, and only if,
\[ X(S_k) \in \mathrm{cl}\left(\mathrm{co}\left(X (\Sigma_k(\se))\right)\right).\]
\end{enumerate}
\end{proposition}

\begin{theorem}[Arbitrage-free invariance] \label{teo:preservaarb2}
Assume $f: \mathrm{dom}~X \rightarrow \mathrm{dom}~X'$ to be a map satisfying (\ref{scalar}) and that the function $F$, induced by Lemma \ref{lem:diagconm}, is strict inversely convexity preserving. Given a trajectory based market $\Me= \Se \times \He$, let $\se \in \Se$ and $k \ge 0$. If $(\se,k)$ is an arbitrage-free node, then $(\se',k)$, where $S'_i=f(S_i)~~i\ge 0$, is an arbitrage-free node in a transformed market $\Me'= \Se' \times \He'$, i.e.
\[ 0 \in \mathrm{ri}(\mathrm{co}(\{ X'(f(\hat{S}_{k+1}))-X'(f(S_k)) : \hat{\se} \in \Se_{(\se,k)} \})) \subseteq \mathbb{R}^{d'}\]
and so $f$ is a no-arbitrage symmetry.
\end{theorem}
\begin{proof}
We know from Lemma \ref{lem:diagconm} that there exists $F:\mathrm{dom}~X \to \mathrm{dom}~X'$ given by $F(x)=X'(f(s))$, where $s \in \mathrm{dom}~X$ such that $X(s)=x$. Thus, since by hypothesis it is a strict inversely convexity preserving map, from Lemma \ref{teo:preservaarb} and Lemma \ref{translation}, it follows (by taking $x_0=X(S_k)$) that
   \[0 \in \mathrm{ri}(\mathrm{co}(\{F(x)-F(x_0) : x \in X(\Sigma_k(\se)) \})),\]
or, equivalently,
\[0 \in \mathrm{ri}(\mathrm{co}(\{ X'(f(\hat{S}_{k+1}))-X'(f(S_k)) : \hat{\se} \in \Se_{(\se,k)} \})).\]
\end{proof}

\begin{remark} By Lemma \ref{lem:diagconm} item  1, if $f$ is strict inversely convexity preserving, then the induced $F$ satisfies the hypothesis of Theorem \ref{teo:preservaarb2}. Also notice that if $Im\;F$ contains a nondegenerate triangle, by Theorem \ref{functionInducingISCP}, $f$ is of the form given by (\ref{fractionalf}). 
\end{remark}

\begin{corollary}[Explicit Characterization] \label{characterization}
Assume $f: \mathrm{dom}~X \rightarrow \mathrm{dom}~X'$ satisfies (\ref{scalar}).
\begin{enumerate}
\item If $f$ is a no-arbitrage symmetry (as per Definition \ref{marketTransformation}) for any market and $\mbox{Im}~f$ is not contained in a $2$-dimensional subspace then $f$ is characterized by expression (\ref{fractionalf}).
\item Conversely if $f$ has the form (\ref{fractionalf}) then it is a no-arbitrage symmetry for any possible market.
\end{enumerate}
\end{corollary}
\begin{proof}
We recall that (\ref{scalar}) assures the existence of the induced function $F$ as in Lemma \ref{lem:diagconm}. Assume $f$ is a no-arbitrage symmetry from a market $\Me$ onto a market $\Me'$. Then, by Theorem \ref{teo:preservaarb2} and Lemma \ref{teo:preservaarb}, the induced $F$ must be strict inversely convexity preserving. Moreover, if $\mbox{Im}~f$ is not contained in a $2$-dimensional subspace Lemma \ref{triangle} implies that $\mbox{Im}~F$ is not contained in a straight line. Finally by Theorem \ref{functionInducingISCP} $f$ takes the form (\ref{fractionalf}). This proves 1.

For the converse, if $f$ has the form (\ref{fractionalf}), the converse of Theorem \ref{functionInducingISCP} implies that the induced function $F$ is strict inversely convexity preserving. Thus by Lemma \ref{teo:preservaarb} and
Theorem \ref{teo:preservaarb2} $f$ is a no-arbitrage symmetry for any possible market.\end{proof}

\vspace{.05in}
Observe that the composition of no-arbitrage symmetries of the form (\ref{fractionalf}) is again of this form.

A transformation of interest in financial terms is the one that changes the market model's
num\'eraire. Let's assume that the first asset $S^1$ is strictly positive for every trajectory in $\Se$, so the first coordinate can take the place of an alternative
num\'eraire for the model. For each $\se \in \Se$, we will denote the sequence of prices relative to $S^1$ by $Y(\se)=\{(Y(S_i),W_i,m)\}_{i \ge 0}$  where $Y: D'\subset\mathbb{R}^{d+1} \rightarrow \mathbb{R}^d$ is the perspective function over the second coordinate:
\begin{equation} \label{eqn:cambionumeraire}
 Y(s) \equiv \left(\frac{s^{0}}{s^1},\frac{s^{2}}{s^1},\dots,\frac{s^{d}}{s^1}\right) \quad D'\equiv\{ s=(s^0,\dots,s^d) \in \mathbb{R}^{d+1}:s^1>0\}.
 \end{equation}
$Y^j(S_i)$ represents the value of the $ j $-th asset in units of the new  num\'eraire. We will prove next the following proposition that will be useful for the coming results.

\begin{proposition} \label{lem:comp}
  Let $\sigma$ be the permutation on $\mathbb{R}^{d+1}$ that interchanges the first coordinate with the second and $X$ the perspective function on $\mathbb{R}^{d+1}$ (defined in \eqref{eqn:X}). Then, $Y = X \circ \sigma$ over $D" \equiv \{s \in \mathrm{dom}~X : s^1>0 \}$. Furthermore, $\sigma$ is a strict segment preserving map.
\end{proposition}
\begin{proof}
  Fix $s \in D"$, then $\sigma(s) \in \mathrm{dom}~X$ and
\begin{eqnarray*}
  (X\circ \sigma) (s) = X\left(s^1,s^0,\dots,s^d\right)=\left(\frac{s^0}{s^1},\dots,\frac{s^d}{s^1}\right) = Y(s).
\end{eqnarray*}
Since $\sigma$ is a linear map, it follows from Theorem \ref{prop:afin}
in Appendix \ref{app:conv}, that it is a strict segment preserving map.
\end{proof}

We are now in a position to show that the arbitrage-free condition on a trajectory based market $\Me$ is independent of the choice of num\'eraire. For this, we will state the following Corollary.

\begin{corollary} \label{cor:numeraire}
Let $\Me=\Se \times \He$ a semi-bounded trajectory based market such that $S^1 >0$ for all $\se \in \Se$ and $\He$ contains the class of restricted portfolios (\ref{restrictedPort}). If $\Me$ is arbitrage-free with $S^0$ as num\'eraire, then $\Me$ is arbitrage-free with $S^1$ as num\'eraire.
\end{corollary}
\begin{proof}
From Proposition \ref{lem:comp} above, it follows that $Y = X\circ \sigma$ on the set $D" \equiv \{s \in \mathrm{dom}~X : s^1>0 \}$,
 where $\sigma$ is the permutation of the first coordinate by the second. Also, since $\Me$ is arbitrage-free, it follows from Proposition \ref{globalImpliesLocalNA} that $\Me$ is locally arbitrage-free. Then, any node $(\se,k)$ in the market is arbitrage-free (all notions with $S^0$ as num\'eraire). As $\sigma$ verifies the hypothesis of Theorem \ref{teo:preservaarb2}, we can ensure that
\[0 \in  \mathrm{ri}\left(\mathrm{co}\left(\Delta Y(\Se_{(\se,k)})\right)\right) \equiv \mathrm{ri}\left(\mathrm{co}\left(\{ Y(\hat{S}_{k+1})-Y(S_k): \hat{\se} \in \Se_{(\se,k)} \} \right) \right),\]
for all $(\se,k)$, or, in other words, every node is arbitrage-free with respect to the num\'eraire $S^1$. Then, it follows from Theorem \ref{teo:locarbfree}, that $\Me$ is arbitrage-free with respect to the num\'eraire $S^1$.
\end{proof}

\vspace{.2in}
Our goal is now to find market transformations $f$ that preserve $0$-neutral nodes (i.e. $0$-neutral symmetries). From Lemma \ref{lem:preserva0neutral} we know that the induced transformation $F$ needs to be continuous and inversely convexity preserving in order to preserve the closure of convex sets. The following Theorem shows that these conditions are, somehow, sufficient to obtain a $0$-neutral symmetry as per Definition \ref{marketTransformation}.

\begin{theorem}[$0$-neutral invariance] \label{teo:preserva0neutral}
Let $f: \mathrm{dom}~X \rightarrow \mathrm{dom}~X'$ be a continuous map satisfying (\ref{scalar}) and the function $F$, induced by Lemma \ref{lem:diagconm}, is inversely convexity preserving. Given a trajectory based market $\Me= \Se \times \He$, let $\se \in \Se$ and $k \ge 0$. If $(\se,k)$ is a $0$-neutral node, then $(\se',k)$ is a $0$-neutral node in the transformed market $\Me'= \Se' \times \He'$, i.e.
\[
X'(f(S_k))\in\mathrm{cl}\left(\mathrm{co}\left(X'(f(\Se_{(\se,k)})) \right)\right)
\]
and so $f$ is a $0$-neutral symmetry.
\end{theorem}
\begin{proof}
We know from Lemma \ref{lem:diagconm}, in Section \ref{sec:geometryPreserving}, that there exists a continuous map $F:\mathbb{R}^d \to \mathbb{R}^{d'}$ given by $F(x)=X'(f(s))$, where $s \in \mathrm{dom}~X$ such that $X(s)=x$. Moreover, by hypothesis it is inversely convexity preserving.

Thus, since by hypothesis $X(S_k)\in \mathrm{cl}\left(\mathrm{co}\left(X(\Se_{(\se,k)}) \right)\right)$, from Lemma \ref{lem:preserva0neutral} in Section \ref{sec:geometryPreserving}, it follows that,
 \[X'(f(S_k)) \in \mathrm{cl}\left( \mathrm{co}\left(X'(f(\Se_{(\se,k)}))\right)\right).\]
\end{proof}

By the converse of Theorem \ref{functionInducingISCP}, if $f$ is given as in the expression (\ref{fractionalf}), with the
prescribed conditions,
then the induced funtion $F$  has the expression (\ref{palesFormDerived}). Therefore, it is also
inversely convexity preserving and continuous by
Theorem \ref{prop:afin}; then $f$ preserves $0$-neutral nodes and so it is $0$-neutral symmetry.

\vspace{.1in}
In the $0$-neutral market definition, the selection of an explicit num\'eraire is required. Consider, as in Corollary \ref{cor:numeraire}, a trajectory based markets such that $S^1> 0$ for all trajectories, this will allow to take that coordinate as an alternative num\'eraire.

\begin{corollary}
Given a semi bounded trajectory based market $\Me= \Se \times \He$ such that $S^1 >0$ for all $\se \in \Se$. If $(\se,k)$ is a $0$-neutral node with respect to the num\'eraire $S^0$, then it is also with respect to the num\'eraire $S^1$. In particular, if $\Me$ is locally $0$-neutral, it will also be $0$-neutral for any choice of num\'eraire.
\end{corollary}
\begin{proof}
From Proposition \ref{lem:comp} it follows $Y = X \circ \sigma$ over the set $D" = \{s \in \mathrm{dom}~X : s^1>0 \}$, where $\sigma$ is the permutation of the first coordinate by the second. Since $\sigma$ verifies the hypothesis of Theorem \ref{teo:preserva0neutral}, then
\[0 \in  \mathrm{cl} \left( \mathrm{co}\left(\left\lbrace Y(\hat{S}_{k+1})-Y(S_k) : \hat{\se} \in \Se_{(\se,k)} \right\rbrace \right)\right)\]
 and $(\se,k)$ is a $0$-neutral node with respect to the num\'eraire $S^1$. We can conclude that if $\Me$ is locally $0$-neutral with respect to the num\'eraire $S^0$, then it will also be $0$-neutral with respect to the num\'eraire $S^1$. Therefore, it follows from Theorem \ref{teo:0neutral}, that if $\Me$ is locally $0$-neutral for $S^0$, $\Me$ will be $0$-neutral for any other choice of num\'eraire.
\end{proof}

\section{Example}
\label{example}

We will provide a slightly non-traditional development on the call-put parity relationship. This is a simple relation among prices of certain assets; it is derived in many  textbooks and can be obtained through a no-arbitrage based proof. We will derive it under the weaker hypothesis of $0$-neutrality  and relate the relationship to NAS (No-Arbitrage Symmetries). Our main point of revisiting the call-put parity is that it will allow us to provide an explicit example of NAS (besides a change of numeraire) as well as to illustrate their meaning in this context.

\subsection{Call-Put Parity Under $0$-Neutrality}
Consider an arbitrary time evolution of four assets
$S_t \equiv (C_t, P_t, Y_t, B_t),$ $0\leq t \leq T$. We require,
\begin{equation} \nonumber
C_T= (Y_T- B_T)_+, P_T= (B_T- Y_T)_+,~\mbox{and}~B_T= K~\mbox{where ~$K$~is a constant}.
\end{equation}
That is: $C$ is a European call written on asset $Y$, with strike $K$ and expiration $T$. Similarly for the European put $P$.  $B$ is a bond. Clearly  $ (C_T-P_T-Y_T+B_T)=0$,
which can be thought as a boundary condition. Under an appropriate no-arbitrage assumption the call-put parity is the following result \cite[Cor 1.4.2]{musiela}:
\begin{equation} \label{noArbitrageConstraint}
(C_t-P_t-Y_t+B_t)=0,~\forall~~0 \leq t \leq T.
\end{equation}
That is, no-arbitrage, under the said conditions, constraints the evolution of the four assets accordingly to (\ref{noArbitrageConstraint}).

 We will add details on dimensions that are neglected in the above formulation, dispensing with units/dimensions is standard in the literature but making them explicit is
relevant to our philosophy as a change of units should be a NAS (but we do not explore this view in the paper). We will insert appropriate dimensions/units whenever relevant but switch (or alternate) to suppressing units (as usual) whenever the relevant dimensions have been made clear.
We write $Z= (Z) [Z]$ where $(Z)$ is the (dimensionless)
numerical value and $[Z]$ the dimensional units of the variable $Z$ respectively.

We will have $[C_t] = \frac{{\bf 1}_{\$}}{{\bf 1}_{C}}$ which would require (see below)  the insertion of  a dimensional constant
$a$ with units $[a]= \frac{{\bf 1}_{S}}{{\bf 1}_{C}}$ with $(a)$ representing
the number of shares associated to a call option. We will take  $(a)=1$ but an arbitrary value of $(a)$ will have the effect of multiplying the call-put parity by $(a)$ (usually, in practice $(a)=100$). So $a$
represents the number of shares per call contract, this is not an artificial insertion as it is a feature of traded options. Similarly $P_t$ will contain a dimensional constant $b$ with $[b]= \frac{{\bf 1}_{S}}{{\bf 1}_{P}}$ 
with $(b)$ representing
the number of shares associated to a put option, we will take $(b)= (a)=1$ in order to derive the put-call parity relationship (as indicated, one can multiply the resulting expression by an arbitrary dimensionless number $(a)$).

In order to provide a derivation of (\ref{noArbitrageConstraint}) under $0$-neutrality, we first express the above setting in our trajectorial framework. The above formulation is in continuous time but we consider this to be a nonessential point (as we argue below). We will work with trajectories of the form
${\bf S}_i= (S_i, t_i, m) = (S_i^0, S_i^1, S_i^2, S_i^3,  t_i, m)=  (C_i, P_i, Y_i, B_i, t_i, m)$
where $0= t_0 < t_1< \ldots < t_m= T$.
Clearly, the times $t_i$ are trajectory dependent; as a particular case we could take $t_i= \frac{i ~T}{M}$, $0 \leq i \leq M$ for a given constant $M$.  Given that the argument will apply to any trajectory set with these coordinates we can approximate any arbitrary time $t$ by taking $M$ larger.

Let $\mathcal{S}$ denote any $0$-neutral trajectory set with the above introduced coordinates and that obeys
\begin{equation}  \nonumber
S^0_{M({\bf S})}=C_{M({\bf S})}= a~(S^2_{M({\bf S})}- [K]~S^3_{M({\bf S}))})_+=
a~(Y_{M({\bf S})}- K~ \frac{{\bf 1}_{\$}}{{\bf 1}_{B}})_+,
\end{equation}
\begin{equation} \nonumber
S^1_{M({\bf S})}= P_{M({\bf S})}= b~([K]~S^3_{M({\bf S})}- S^2_{M({\bf S}))})_+=
b~(K~\frac{{\bf 1}_{\$}}{{\bf 1}_{B}}- Y_{M({\bf S})})_+,
\end{equation}
\begin{equation} \nonumber
~\mbox{and}~S^3_{M({\bf S})}= B_{M({\bf S})}= (K)~\frac{{\bf 1}_{\$}}{{\bf 1}_{B}}~\mbox{for all}~{\bf S}. ~K~\mbox{is a dimensional constant with}~ [K]= \frac{{\bf 1}_{B}}{{\bf 1}_{S}}, 
\end{equation}
$K$ represents the number of bond units per share and so $K~ \frac{{\bf 1}_{\$}}{{\bf 1}_{B}}$ is the strike price.
So we have
$[C_i] = [a]~\frac{{\bf 1}_{\$}}{{\bf 1}_{Y}}$, $[P_i] = [b]~ \frac{{\bf 1}_{\$}}{{\bf 1}_{Y}}$, $[Y_i] = \frac{{\bf 1}_{\$}}{{\bf 1}_{Y}}$ and
$[B_i] = \frac{{\bf 1}_{\$}}{{\bf 1}_{B}}$.  Moreover, assume $M({\bf S})= m$ to be a stopping time in the sense
that if ${\bf S}'_k= {\bf S}_k$ for all $0 \leq k \leq M({\bf S})$ then
$M({\bf S}')= M({\bf S})$. Finally, we also assume $t_{M({\bf S})}=T$. Such $\mathcal{S}$  will be called {\it admissible}.

\noindent
The previous call-put parity is now written with units and taking $(a)= (b)$:
\begin{equation} \label{noArbitrageConstraint2}
\beta(S_i)\equiv ( {\bf 1}_C~C_i-{\bf 1}_P~ P_i -{\bf 1}_Y~Y_i+ {\bf 1}_B~B_i )=0,~\forall~~0 \leq i \leq M.
\end{equation}
That is, no-arbitrage, under the said conditions, constraints the evolution of the four assets accordingly to (\ref{noArbitrageConstraint2}).
(\ref{noArbitrageConstraint2}) holds if and only if 
${\bf 1}_B~ \pi(X(S_i)) \equiv  {\bf 1}_B [(\frac{C_i}{B_i})- (\frac{P_i}{B_i})- (\frac{Y_i}{B_i}) +1)]=0$ where, as defined before  $X(S_{M({\bf S})})= (\frac{S^0_{M({\bf S})}}{S^3_{M({\bf S})}},\frac{S^1_{M({\bf S})}}{S^3_{M({\bf S})}}, \frac{S^2_{M({\bf S})}}{S^3_{M({\bf S})}})= (\frac{C_{M({\bf S})}}{B_{M({\bf S})}},\frac{P_{M({\bf S})}}{B_{M({\bf S})}}, \frac{Y_{M({\bf S})}}{B_{M({\bf S})}})$ (notice that we are abusing the notation by using
$S^3$ as numeraire instead of the usual $S^0$).

\noindent
To establish (\ref{noArbitrageConstraint2}), we will return now to the usual practice of suppressing the units, in particular, in the proof below when we write $X(S_i)$ it will be interpreted
as the coordinates without the dimensions i.e. $((\frac{C_i}{B_i}),(\frac{P_i}{B_i}), (\frac{Y_i}{B_i}))$.

\subsection{Proof of Call-Put Parity}

Let $\Pi \equiv \{x \in \mathbb{R}^3: \pi(x^1, x^2, x^3)= x^1-x^2-x^3+1=0\}$. Consider an admissible trajectory set as described above; according to Proposition 5.2 item 2: $X(S_{M({\bf S})-1}) \in \mbox{cl}(\mbox{co}(X(\Sigma_{M({\bf S})-1}({\bf S}))))$. Clearly  $\mbox{cl}(\mbox{co}(X(\Sigma_{M({\bf S})-1}({\bf S})))) \subseteq \Pi$ and therefore
 $\pi(X(S_{M({\bf S})-1})) =0$.  Continuing the argument by induction we obtain
  $\pi(X(S_i))= [(\frac{C_i}{B_i})- (\frac{P_i}{B_i}) - (\frac{Y_i}{B_i})+1] =0$ for all $0 \leq i \leq M({\bf S})$ which is our version of the call-put parity. The result is here established solely under the hypothesis of $0$-neutrality that is weaker than the no-arbitrage assumption.

\subsection{An Example of a NAS}

Let us introduce the following transformation:
\begin{equation} \nonumber
C_i \rightarrow C_i'= \frac{P_i}{Y_i~B_i}, ~P_i \rightarrow P_i'= \frac{C_i}{Y_i~B_i}, Y_i \rightarrow Y_i'= \frac{1}{Y_i}, ~B_i \rightarrow B_i'= \frac{1}{B_i}.
\end{equation}
So
\begin{equation} \nonumber
(C_i, P_i, Y_i, B_i) \rightarrow (C'_i, P'_i, Y'_i, B'_i)= \frac{1}{Y_i~B_i}(P_i, C_i, B_i, Y_i, ).
\end{equation}
We then have (we are disregarding dimensional constants with numerical value $1$):
\begin{equation} \nonumber
C'_{M({\bf S})}= (Y'_{M({\bf S})}- B'_{M({\bf S})})_+= (\frac{1}{Y_{M({\bf S})}}- \frac{1}{K})_+
\end{equation}
\begin{equation} \nonumber
P'_{M({\bf S})}= (B'_{M({\bf S})}- Y'_{M({\bf S})})_+= (\frac{1}{K}- \frac{1}{Y_{M({\bf S})}})_+.
\end{equation}

In financial terms, the transformed variables $C'_i, P'_i$ are prices of a call and a put options, respectively, but now depending on the price of the same asset $Y_i$  but expressed in terms of shares per currency unit. This is not equivalent to using $Y$ as the numeraire. 

Notice that 
\begin{equation} \nonumber
C'_i- P'_i- Y'_i+ B'_i=
 \frac{1}{Y_{i}~B_{i}}(P_{i}-C_{i}-B_{i}+Y_{i})=0.
\end{equation}

In fact, we  will argue that $\rightarrow$ is indeed a NAS. We change notation to 
touch basis with the formal notation in the  paper, let: 
$s \rightarrow s'$ being given by $s'= f(s)$. 
where, with the notation $s \equiv (s^0, s^1, s^2, s^3)$,
 $~~f(s^0, s^1, s^2, s^3) = \frac{(s^1, s^0, s^3, s^2)}{s^2 s^3}$ and notice that if $L(s^0, s^1, s^2, s^3) \equiv (s^1, s^0, s^3, s^2)$, a linear function,  
we obtain:
$f(s)= \frac{f^3(s)}{L^3(s)} L(s)$. So $f$ has the form (\ref{fractionalf})
and by Corollary \ref{characterization} item $2$,  $f$ is a no-arbitrage symmetry. In fact, $f$ preserves $0$-neutrality as well and this follows from the converse of Theorem \ref{functionInducingISCP} and Theorem  \ref{teo:preserva0neutral}.

\subsection{Call-Put Parity Under a No-Arbitrage Symmetry}
Let us now see the effect on the call-put parity relation after applying a no-arbitrage symmetry. Towards this goal, consider $f$ to be a no-arbitrage symmetry satisfying (\ref{scalar}) and such that $Im f$ is not contained in a $2$-dimensional subspace.
From Corollary \ref{characterization}, we have $f(S_i^0, \ldots, S_i^4)= f(S_i)
= \frac{f^3(S_i)}{L^3(S_i)} L(S_i)$ where $f^3, L^3, L$ are as in Theorem \ref{functionInducingISCP}. Then $$F(X(S_i))= F(\frac{S^0_i}{S^3_i},\frac{S^1_i}{S^3_i},\frac{S^2_i}{S^3_i})=
(\frac{C_i'}{B_i'},\frac{P_i'}{B_i'},\frac{Y_i'}{B_i'})=
(\frac{L^0(S_i)}{L^3(S_i)},\frac{L^1(S_i)}{L^3(S_i)},\frac{L^2(S_i)}{L^3(S_i)}).$$ All in all, we will then take (with some abuse of notation):
\begin{equation} \nonumber
F(x)= \frac{A(x)+b}{B(x) +c}~~\mbox{where}~(B(x) +c) >0,
\end{equation}
with $A:\mathbb{R}^3 \rightarrow \mathbb{R}^3$ and  $B:\mathbb{R}^3 \rightarrow \mathbb{R}$ both linear transformations (notice that we have reproduced computations from Theorem \ref{functionInducingISCP}).

Before proceeding to a computation we need to impose that the boundary condition behaves as follows:
\begin{equation} \label{transformedPricesAreOptionPrices}
C'_T= (Y'_T- B'_T)_+, P'_T= (B'_T- Y'_T)_+,
\end{equation}
that is, the corresponding transformed price coordinates are prices
of a call and a put on the transformed asset. Such an imposition is necessary for the derivation to follow and prescribes that the boundary condition is invariant under $F$.

We briefly sketch an argument establishing
\begin{equation} \label{putCallParityIsInvariant}
\pi(F(X(S_i))) = \frac{a_F}{B(x) +c} \pi(X(S_i)),
\end{equation}
where $a_F \equiv  (a_{1,1}- a_{2,1} - a_{3,1}- a_{4,1})$ and $a_{j,k}$ are
the matrix coordinates of a matrix representation of $A$. The relationship
(\ref{putCallParityIsInvariant}) makes it immediate that $\pi(X(S_i))=0$
implies $\pi(F(X(S_i)))=0$ and hence reflecting the notion of symmetry embodied
by $F$. The implication  $\pi(X(S_i))=0 \implies \pi(F(X(S_i)))=0$  is known to us
without recourse to (\ref{putCallParityIsInvariant}), this is so because
$f$ is a no-arbitrage symmetry and so a $0$-neutral symmetry and
given that $\mathcal{S}$ is assumed to be $0$-neutral so will then be
$\mathcal{S}'$ (this trajectory set obtained from $\mathcal{S}$ by acting
with $f$ on the trajectories ${\bf S} \in \mathcal{S}$).

Given that $\pi$ is linear it is enough to consider the case $F(x) = A(x)$ and establish the existence of $a_F$ such that $\pi(F(X(S_i))) = a_F ~\pi(X(S_i))$.

To start,  substracting the two equations in (\ref{transformedPricesAreOptionPrices}) we obtain:
\begin{equation} \label{usefulRelationship}
(a_{1,1}- a_{2,1}) C_T+ (a_{1,2}- a_{2,2}) P_T+ (a_{1,3}- a_{2,3}) Y_T
+(a_{1,4}- a_{2,4}) B_T=
\end{equation}
\begin{equation} \nonumber
(a_{3,1}- a_{4,1}) C_T+ (a_{3,2}- a_{4,2}) P_T+ (a_{3,3}- a_{4,3}) Y_T
+(a_{3,4}- a_{4,4}) B_T.
\end{equation}
It turns out, that in order to establish $\pi(F(X(S_i))) = a_F ~\pi(X(S_i))$, we will only need to obtain some relationships among the matrix entries $a_{i,j}$. For reasons of space we only sketch the derivations which follow from (\ref{usefulRelationship}). First, let $Y_T > B_T$
and  equating
coefficients of $Y_T$ (equating coefficients of variables does require some minimal assumptions on $Y_T$ and $B_T$ which we do not make explicit) we obtain
\begin{equation}  \label{sTTermCaseI}
(a_{1,1}- a_{2,1}) + (a_{1,3}- a_{2,3})= (a_{3,1}- a_{4,1})+ (a_{3,3}- a_{4,3}),
\end{equation}
a similar relation is obtained for the coefficients of $B_T$.
Two more analogous relationships among coefficients are obtained from the case $Y_T < B_T$. The said relationships allow to evaluate as follows
\begin{eqnarray*}
\pi(F(X(S_i))) &=&
\pi(A(S_i))= a_{1,1} C_i+ a_{1,2} P_i+ a_{1,3} S_i+ a_{1,4} B_i\\
&-& a_{2,1} C_i- a_{2,2} P_i- a_{2,3} S_i- a_{2,4} B_i
- a_{3,1} C_i- a_{3,2} P_i\\
&-& a_{3,3} S_i- a_{3,4} B_i
+ a_{4,1} C_i+ a_{4,2} P_i+ a_{4,3} S_i+ a_{4,4} B_i\\
&=& (a_{1,1}-a_{2,1}- a_{3,1}+a_{4,1})C_i+  (a_{1,2}-a_{2,2}- a_{3,2}+a_{4,2})P_i\\
&+& (a_{1,3}-a_{2,3}- a_{3,3}+a_{4,3}) S_i+  (a_{1,4}-a_{2,4}- a_{3,4}+a_{4,4})B_i\\
&=& a_F ~\pi(X(S_i)).
\end{eqnarray*}

\section{Conclusion}

The paper poses and solves the following basic question: what transformations, acting on financial events,
leave the no-arbitrage property invariant? Such transformations are called no-arbitrage symmetries (NAS) and are interpreted as  mapping financial events to financial events. We make use of results from convex analysis and a general non-probabilistic framework to characterize and provide explicit expressions for the NAS. We take advantage of a formulation of arbitrage free markets (as per Section \ref{geometricCharacterizations}) in terms of geometric assumptions of the trajectories in discrete time. The problem formulation naturally provides  the characterization, in a local sense, of no-arbitrage preserving transformations.

The transformed variables, i.e. the output values of NAS, do require an interpretation as the original setting is abstract and general. For example, in the example of Section \ref{example} we have to impose that boundary conditions should also be invariant under NAS and in so doing we required that two of the transformed variables acted as call and put options on the two remaining transformed variables. From such a general point of view we think that the result of applying a NAS to financial events are admissible prices for financial events but the latter will require an interpretation that will depend on the context and the specific NAS under consideration.

\appendix

\section{Results and Proofs from Section \ref{sec:arby0neutral}}
\label{proofsAndResults}

The following simple characterization of $0$-neutral markets is used in one of ours results.
\begin{proposition} \label{prop:epcons}
A trajectory based market $\Me=\Se \times \He$ is $0$-neutral if and only if, for each $\Phi \in \He$ and $\epsilon >0$ there exist $\se^{\epsilon} \in \Se$ such that
\begin{equation} \label{eqn:epcons}
\sum_{i=0}^{N_{\Phi}(\se^{\epsilon})-1} H_i(\se^{\epsilon}) \cdot \Delta_i X(S^{\epsilon})<\epsilon.
\end{equation}
\end{proposition}
\begin{proof}
  Suppose first $\Me$ is $0$-neutral. From the definition follows that for any $\epsilon >0$
  \[ \inf_{\se' \in \Se} \left[ \sum_{i=0}^{N_{\Phi}(\se')-1} H_i(\se') \cdot \Delta_i X(S') \right] \le 0 < \epsilon \]
  for all $\Phi \in \He$. Then, for each $\Phi$ there exist $\se^{\Phi} \in \Se$ such that
  \[ \sum_{i=k}^{N_{\Phi}(\se^{\Phi})-1} H_i(\se^{\Phi}) \cdot \Delta_i X(S^{\Phi}) < \epsilon \]
  for any $\epsilon>0$. Thus we proved the necessary condition.

  For the sufficient condition, fix $\epsilon >0$, then, by hypothesis, for each $\Phi \in \He$ there is $\se^{\epsilon} \in \Se$ such that
   \[ \sum_{i=0}^{N_{\Phi}(\se^{\epsilon})-1} H_i(\se^{\epsilon}) \cdot \Delta_i X(S^{\epsilon})<\epsilon.\]
   Then, for each $\Phi \in \He$
   \[  \inf_{\se' \in \Se} \left[ \sum_{i=0}^{N_{\Phi}(\se')-1} H_i(\se') \cdot \Delta_i X(S') \right] < \epsilon.\]
   Since $\epsilon>0$ was chosen arbitrarily, it follows that
   \[  \inf_{\se' \in \Se} \left[ \sum_{i=0}^{N_{\Phi}(\se')-1} H_i(\se') \cdot \Delta_i X(S') \right] \le 0 \]
   for all $\Phi \in \He$. Therefore, since $0 \in \He$, we conclude that $\Me$ is $0$-neutral.
\end{proof}

\noindent
{\bf Proof of Theorem \ref{teo:locarbfree}}

\begin{proof}
Assume $\mathcal{M}$ is locally arbitrage-free and  semi-bounded;  fix $\Phi \in \He$ once and for all.  If for all nodes $(\se,k)$, $H_k(\se) \cdot \Delta_k X(S') = 0$  holds  for all $\se' \in \Se_{(\se,k)}$ then
\[G^{\Phi}_{N_{\Phi}}(\se)=\sum_{i=0}^{N_{\Phi}(\se)-1} H_i(\se) \cdot \Delta_i X(S) =0  \]
for all $\se \in \Se$ and so
\[V^{\Phi}_{N_{\Phi}}(\se)=V^{\Phi}_0+G^{\Phi}_{N_{\Phi}}(\se)=V^{\Phi}_0, ~\forall \se \in \Se; \]
therefore,  $\Phi$  is not an arbitrage opportunity.

We may then assume that there exists a trajectory $ \se^{(0)} \in \Se$ and an integer $k \geq 0$ such that at the node $(\se^{(0)},k)$, $H_k(\se^{(0)}) \cdot \Delta_k X(S) \neq 0$ for some $\se \in \Se_{(\se^{(0)},k)}$.

Then, by Definition \ref{localWithRespectToH}, 1.,
it is possible to choose $k_1$, $0 \leq k_1 \leq k$, the smallest integer such that, for $0 \leq j < k_1$, $H_j(\se) \cdot \Delta_j X(\se) =0$ for all $\se \in \Se_{(\se^{(0)},j)}$,
and there exists $\se^{(1)} \in (\se^{(0)},k_1)$ such that
\[ \sum_{i=0}^{k_1} H_i(\se^{(1)}) \cdot \Delta_i X(S^{(1)})<0.\]
Consider the case when for all $k_1<k\le N_{\Phi}(\se^{(1)})$,
$H_k(\se^{(1)}) \cdot \Delta_k X(S^{(1)}) = 0$ holds
(such case we label $(\ast$)); then
\[G^{\Phi}_{N_{\Phi}}(\se^{(1)})=\sum_{i=0}^{N_{\Phi}(\se^{(1)})-1} H_i(\se^{(1)}) \cdot \Delta_i X(S^{(1)}) <0, \]
under condition ($\ast$) we have then established that $\Phi$ is not an arbitrage opportunity.

Otherwise, i.e. when the case $(\ast)$ does not hold, we proceed by induction.
Assume that for $i\ge 1$ it was obtained the strictly increasing sequence of non negative integers $(k_j)_{j=1}^i$ and $\se^{(j)} \in \mathcal{S}_{(\se^{(j-1)}, k_j)},\; 1\le j\le i$, such that for $k_{j-1} < k < k_j,\;(k_0=0)$, $H_k(\se^{(j)}) \cdot \Delta_k X(S^{(j)})=0$, and
$H_{k_j}({\bf S}^{(j)}) \cdot \Delta_{k_j} X(S^{(j)}) <0$. In particular
\begin{equation}\nonumber
\sum_{j=0}^{k_i} H_j(\se^i) \cdot \Delta_j X(S^i) <0.
\end{equation}
The same argument that we used for the node  $(\se^{(1)},k_1)$ above, but now applied to $(\se^{(i)},k_i)$, and the inductive hypothesis gives the logical alternatives:

$a)$ $\Phi$ is not an arbitrage opportunity by condition (*),

$b)$ the inductive hypothesis holds for $i+1$.

\noindent Due to our hypothesis that $\mathcal{M}$ is semi-bounded and that $\Phi$ is fixed, we remark that the alternative $b)$ becomes, eventually, empty and so the alternative $a)$ holds for $i$ large enough. Since $\Phi$ is arbitrary, $\mathcal{M}$ is arbitrage free.
\end{proof}

\noindent
{\bf Proof of Proposition \ref{globalImpliesLocalNA}}

\begin{proof}
We proceed by contrapositive. Assume $\mathcal{S}$ is not locally arbitrage-free. Therefore, there is a node $({\bf S}, k)$ which is not arbitrage-free, i.e. by Proposition \ref{localDefinitions} (in subsection \ref{local geometric section}), $\Delta X(\Se_{(\se,k)})$ is disperse, so there exists $\xi \in \mathbb{R}^d$
such that
\begin{itemize}
\item $\xi \cdot \Delta_k X(S') \ge 0$ for all $\se' \in \Se_{(\se,k)}$, and
\item there exists $\se^* \in \Se_{(\se,k)}$ such that $\xi \cdot \Delta_k X(S^*) > 0$.
\end{itemize}
Since by hypothesis $\Xi^{(\se,k)}$ belongs to $\He$, it follows from Proposition \ref{prop:auto} that
\[V^{\Xi^{(\se,k)}}_{N_{\Xi^{(\se,k)}}}(\se')= \xi \cdot \Delta_k X(S') \ge 0,\]
for all $\se' \in \Se$ and there exists $\se^* \in \Se$ such that
\[V^{\Xi^{(\se,k)}}_{N_{\Xi^{(\se,k)}}}(\se^*)= \xi \cdot \Delta_k X(S^*) > 0.\]
Therefore, $\Xi^{(\se,k)}$ is an arbitrage opportunity.
\end{proof}

\noindent
{\bf Proof of Theorem \ref{teo:0neutral}}

\begin{proof} Fix $\Phi \in \He$ and $\epsilon >0$. We are going to show that there exists $\se^{\epsilon}\in\Se$ such that (\ref{eqn:epcons}) holds.

Fix $\se\in\Se$, given that $(\se,0)$ is a $0$-neutral node w.r.t.~$\He$, it follows that there exists $\se^{(1)} \in \Se=\Se_{(\se,0)}$ such that $H_0(\se) \cdot \Delta_0 X(S^{(1)}) <\frac{\epsilon}{2}$. \\ Then, if $N_{\Phi}(\se^{(1)})=1$,
    \[ \sum_{i=0}^{N_{\Phi}(\se^{(1)})-1}H_i(\se^{(1)}) \cdot \Delta_i X(S^{(1)}) < \frac{\epsilon}{2} < \epsilon. \]
If $N_{\Phi}(\se^{(1)})>1$, in the same way than before, we can choose a finite sequence $(\se^{(j)})_{j=1}^n$ with $n\le n_{\Phi}$ such that for $2\le j \le n$,
 \[
 \se^{(j)} \in \Se_{(\se^{(j-1)},j-1)}\;\;\mbox{and}\;\;\sum_{i=0}^{j-1}H_i(\se^{(j)}) \cdot \Delta_i X(S^{(j)}) < \sum_{i=1}^{j} \frac{\epsilon}{2^i} < \epsilon.
 \]
Since $\Me$ is semi-bounded, there exists $0\le n \le n_{\Phi}$ such that
    \[ \sum_{i=0}^{N_{\Phi}(\se^{(n)})-1}H_i(\se^{(n)}) \cdot \Delta_i X(S^{(n)}) < \sum_{i=0}^{n} \frac{\epsilon}{2^i} < \epsilon. \]
So $(\ref{eqn:epcons})$ holds with  $\se^{\epsilon}= \se^{(n)}$. Thus, since $\Phi \in \He$ was chosen arbitrarily, it follows from Proposition \ref{prop:epcons} that $\Me$ is $0$-neutral.
\end{proof}

\noindent
{\bf Proof of Proposition \ref{0NeutralGlobalImpliesLocal}}

\begin{proof}
Suppose $\Me$ is $0$-neutral but some node $(\se, k)$ is not $0$-neutral, it then follows from Proposition \ref{localDefinitionsII} (in subsection \ref{local geometric section}) that there exists $\xi \in \mathbb{R}^d$ satisfying
\[ \inf_{\se' \in \Se_{(\se,k)}} \xi \cdot \Delta_k X(S')>0 ~\mbox{for all}~ \se' \in \Se_{(\se,k)}. \]

By hypothesis, $\Xi^{(\se,k)} \in \He$ (see the definition preceding Proposition \ref{globalImpliesLocalNA}). Then
\[ \inf_{\se' \in \Se_{(\se,k)}} \left[ \sum_{i=k}^{N_{\Xi^{(\se,k)}}(\se')-1} \Xi_i^{(\se,k)}(\se') \cdot \Delta_k X(S') \right] > 0,\]
which is a contradiction. Therefore $(\se,k)$ is a $0$-neutral node.
\end{proof}

\noindent

\section{Convex Analysis} \label{app:conv}

For $x,y \in \mathbb{R}^d$ we define the closed segment $[x,y]$ and the open segment $(x,y)$ by
\[ [x,y] \equiv \{ tx+(1-t)y : 0 \le t \le 1\} \mbox{ and } (x,y) \equiv \{ tx+(1-t)y : 0 < t < 1\}.\]
To begin, let's remember the notion of relative interior which will be very important in the characterizations of local properties.

\begin{definition}[Relative interior]
Let $E \subset \mathbb{R}^d$ a convex set. The \emph{relative interior} of $E$, that we will denote by $\mathrm{ri}(E)$, is the interior of the set relative to its affine hull, that is,
\[ \mathrm{ri}(E)=\{ x \in E : B(x,r) \cap \mathrm{aff}~E \subseteq E \mbox{ for some }r>0\}.\]
\end{definition}

The following property relates the notions of closure and relative interior.

\begin{proposition}[{\cite[Teorema 6.1]{rock}}] \label{prop:ri2}
Let $E \subset \mathbb{R}^d$ a non empty convex set. Then, for each $x \in \mathrm{ri}(E)$,
\[ \alpha x+(1-\alpha)y \in \mathrm{ri}(E)\]
for all $y \in \mathrm{cl}(E)$ and for all $\alpha \in (0,1]$.
\end{proposition}

The Proposition that follows describes one of the most important properties of the closure and the relative interior of convex sets.

\begin{proposition} \label{prop:clconv}
Let $E \subset \mathbb{R}^d$ a convex set. Then $\mathrm{cl}(E)$ and $\mathrm{ri}(E)$ are convex sets.
\end{proposition}

The following characterizations of the relative interior for convex sets are useful.

\begin{proposition}[{\cite[Corollary 6.4.1]{rock}}] \label{prop:ri}
Let $E \subset \mathbb{R}^d$ a convex set. Then the relative interior of $E$ is the set of all points $x \in E$ such that for all $y \in E$ there exist some $\epsilon >0 $ with
\[ x-\epsilon (y-x) \in E.\]
\end{proposition}
\begin{corollary}\label{cor:ri}
$x \in \mathrm{ri}(E)$ iff for any $b\in E$ there exists $a\in E$ such that $x\in(a,b)$.
\end{corollary}
\begin{proof}
From Proposition \ref{prop:ri}, if $x \in \mathrm{ri}(E)$, for $b \in E$ there exist some $\epsilon >0 $ with
\[
a\equiv x-\epsilon (b-x)\in E,\;\mbox{so}\; x = \frac 1{1+\epsilon}a+\frac \epsilon{1+\epsilon}b\in (a,b).
\]
Conversely if $x=t\,a + (1-t)\,b$, with $t\in(0,1)$, then
\[
x - \frac {1-t}t(b-x) = a \in E.
\]
\end{proof}

In the following results we will describe some operations that preserve convexity. These operations are helpful in determining or establishing when a set is convex.
Given a map $g : \mathbb{R}^d \to \mathbb{R}^{d'}$, we are going to present two properties of preservation of convex sets by $g$ introduced in \cite{pales}.
We say $g$ \emph{preserves convexity} if $g(E)$ is convex for all convex subset $E \subseteq R^{d}$. Analogously, we say that $g^{-1}$ preserves convexity or $g$ is inversely convexity preserving if $g^{-1} (E')$ is convex whenever $E'$ is a convex subset of $g(\mathbb{R}^{d'})$. The following results are immediate.

\begin{proposition} \label{prop:preservaconvexos}
Let $g : \mathbb{R}^d \to \mathbb{R}^{d'}$,
\begin{enumerate}
\item $g$ preserves convexity if and only if $[g(x),g(y)] \subseteq g([x,y])$\\ for all $x,y \in \mathbb{R}^d$.
\item $g$ is inversely convexity preserving if and only if $g([x,y]) \subseteq [g(x),g(y)] $ for all $x,y \in \mathbb{R}^d$.
\end{enumerate}
\end{proposition}

Note that it follows from the previous Proposition that a convexity preserving function which is, at the same time, inversely convexity preserving satisfy $ [g (x), g (y)] = g ([x, y]) $ for all $ x , y \in \mathbb{R}^ d $. This motivates the following definition.

\begin{definition}[Segment preserving] \label{def:preservasegmentos}
We say that a map $g : \mathbb{R}^d \to \mathbb{R}^{d'}$ \emph{preserves segments} if $[g(x),g(y)] = g([x,y])$ for all $x,y \in \mathbb{R}^d$. If $(g(x),g(y)) = g((x,y))$ for all $x,y \in \mathbb{R}^d$ we say that $g$ \emph{preserves segments strictly}.
\end{definition}

Then, $g$ preserves segments if and only if $g$ preserves convexity and preserves convexity inversely. Clearly, if $g$ preserves segments strictly, then preserves segments, the converse, however, may not be valid.

The obvious candidates to be functions that preserve segments strictly are the affine functions. Recall that a function $g: \mathbb{R}^d \rightarrow \mathbb{R}^{d'}$ is affine if it is the sum of a linear function plus a constant, that is, $g(x)=Ax+b$, where $A \in \mathbb{R}^{d \times d'}$ and $b \in \mathbb{R}^{d'}$. There is a larger class of functions which also preserve segments strictly.

\begin{theorem}[{\cite[Thm 1]{pales}}] \label{prop:afin}
  Let $A: \mathbb{R}^d \rightarrow \mathbb{R}^{d'}$ and $B: \mathbb{R}^d \rightarrow \mathbb{R}$ linear functions, $b\in \mathbb{R}^{d'}$ and $c\in \mathbb{R}$. Let $D=\{x\in \mathbb{R}^d: B(x)+c>0\}$, then: $g:D \rightarrow \mathbb{R}^{d'}$ given by
  \begin{equation}\label{ratioFunction}
  g(x)=\frac{A(x)+b}{B(x)+c}
  \end{equation}
  preserves segments strictly.
\end{theorem}

Consider the function $X: \mathbb{R}^{d+1} \rightarrow \mathbb{R}^{d}$ with $\mathrm{dom}~X= \{x \in \mathbb{R} : x>0\} \times \mathbb{R}^d$ defined in (\ref{eqn:X}) by
 \begin{equation} \nonumber
X(x)=\frac{1}{x^0}(x^1,x^2,\dots,x^d).
 \label{eqn:perspectiva}
 \end{equation}
This function called, \emph{perspective function}, scales or normalizes vectors, so the first component is one, and then drops the first component. Since it has the form (\ref{ratioFunction}), then preserves segments strictly.

The following result is key to our analysis.

\begin{theorem}[{\cite[Thm 2]{pales}}]\label{teo:pales} Let $C\subset \mathbb{V}$ be a nonempty convex subset and $g:C \rightarrow \mathbb{V}'$ be a strict inversely convexity preserving function such that $\mbox{Im}\,f$ contains a nondegenerate triangle. Then, there exist $A: \mathbb{V} \rightarrow \mathbb{V}'$, and
$B:\mathbb{V}\rightarrow \mathbb{R}$ linear functions, $b\in \mathbb{V}'$, and $c\in\mathbb{R}$, such that
\[
B(x)+c >0 \quad \mbox{for}\;x\in C,
\]
and
\begin{equation} \label{palesCharacterization}
g(x)=\frac{A(x)+b}{B(x)+c}.
\end{equation}

\end{theorem}
Moreover, by {\cite[Thm 1]{pales}} $g$ preserves segments strictly, this latter notion means that equality holds in (\ref{invConvexPres}).

We will present below the Separation Theorem that we use to prove the Proposition \ref{localDefinitions}.

\begin{theorem}[{\cite[Prop A.1]{follmer}}] \label{teo:separacion}
Suposse $E \subset \mathbb{R}^d$ is a non empty convex set such that $0 \notin E$. Then, there exists $a \in \mathbb{R}^d$ such that $a \cdot x \ge 0$ for all $x \in E$ and $a \cdot x_0 > 0$ for at least one $x_0 \in E$. Furthermore, if $\inf\limits_{x \in E} \| x\|_d >0$, then we can find $a \in \mathbb{R}^d$ such that $\inf_{x \in E} | a \cdot x | >0$.
\end{theorem}

Next we will define the convex hull of a set.

\begin{definition}[Convex hull]
The \emph{convex hull} of a set $E \subset \mathbb{R}^d$, that we will denote by $\mathrm{co}(E)$, is the smallest convex set containing $ E $.
\end{definition}

One of the most important characterizations of the convex hull is the Carath\'eodory Theorem.

\begin{theorem}[Carath\'eodory theorem] \label{carath}
Let $E \subset \mathbb{R}^d$. Then
\[ \mathrm{co}(E)= \left\lbrace \sum_{i=1}^{d+1} \lambda_ix_i : x_i \in E, \lambda_i \ge 0 , \sum_{i=1}^{d+1} \lambda=1 \right\rbrace. \]
\end{theorem}


\section*{Funding}

 The research of S.E. Ferrando is supported in part by an NSERC grant.\\
 
 The research of I.L. Degano and A.L. Gonz\'alez is supported in part by 
 National University of Mar del Plata, Argentina [EXA902/18].

\end{document}